\documentclass[11pt,a4paper]{article}
\usepackage[T1]{fontenc}
\usepackage[USenglish]{babel}
\usepackage[a4paper,margin=1in]{geometry}
\usepackage{amsmath}
\usepackage{amssymb}
\usepackage{amsthm}
\usepackage{mathtools}
\usepackage{graphicx}
\usepackage{xcolor}
\usepackage[hidelinks]{hyperref}
\usepackage{authblk}
\usepackage{algorithm}
\usepackage{algpseudocode}
\usepackage{mdframed}
\usepackage{xpatch}
\usepackage{calc}

\makeatletter
\xpatchcmd{\algorithmic}{\itemsep\z@}{\itemsep=1.2pt}{}{}
\makeatother

\algtext*{EndWhile}
\algtext*{EndIf}
\algtext*{EndFor}

\newtheorem{theorem}{Theorem}
\newtheorem{lemma}[theorem]{Lemma}
\newtheorem{corollary}[theorem]{Corollary}

\newcommand{\eps}{\varepsilon}

\newcommand{\congest}{\textnormal{\textsc{Congest}}}
\newcommand{\local}{\textnormal{\textsc{Local}}}
\newcommand{\e}{\mathbb{E}}
\newcommand{\opt}{\operatorname{OPT}}
\newcommand{\probit}{p}
\newcommand{\stagebase}{p}
\newcommand{\iterbase}{q}
\newcommand{\vcore}{\ensuremath{V_{\mathrm{core}}}}

\title{Approximating Minimum Dominating Set with Few Awake Rounds}

\author[1]{Hongyan Ji\thanks{\texttt{hongyan-ji@uiowa.edu}; \url{https://orcid.org/0000-0002-6610-844X}. Supported in part by the National Science Foundation (NSF) grants CCF-2402835 and CCF-2617608.}}
\author[2]{Shreyas Pai\thanks{\texttt{shreyas@cse.iitm.ac.in}; \url{https://orcid.org/0000-0003-2409-7807}. Shreyas Pai is grateful to the A. Raghunathan Center for Theoretical CS for support.}}
\author[1]{Sriram V. Pemmaraju\thanks{\texttt{sriram-pemmaraju@uiowa.edu}; \url{https://orcid.org/0000-0002-0834-3476}. Supported in part by the National Science Foundation (NSF) grants CCF-2402835 and CCF-2617608.}}
\affil[1]{Department of Computer Science, University of Iowa, Iowa City, IA, USA}
\affil[2]{Department of Computer Science and Engineering, IIT Madras, Chennai, India}
\date{}

\begin{document}
\maketitle

\begin{abstract}
We study the \textit{Minimum Dominating Set (MDS)} problem in the sleeping \congest{} model (Chatterjee, Gmyr, and Pandurangan, PODC 2020), a generalization of the standard \congest{} model, in which a node may sleep in some rounds and can only compute, send messages, or receive messages when it is awake.
The \textit{awake complexity} of an algorithm in this model is the worst case number (over all inputs and all nodes) of rounds a node is awake for during the execution of the algorithm.
While there are several $O(\log \Delta)$-approximation algorithms (in expectation) for MDS that run in $O(\log^2 \Delta)$ rounds, all of these have $\Omega(\log^2 \Delta)$ awake complexity. Whether this awake complexity can be improved is the question that drives our work.
We present the first $O(\log \Delta)$-approximation algorithm for MDS with $O(\log \Delta)$ awake complexity. The approximation guarantee holds in expectation, and the algorithm runs in $O(\log^2\Delta)$ rounds.
We can reduce the awake complexity further, but at the cost of approximation: we present, for any $1<\alpha\le\min\{\Delta,\log n\}$, an algorithm in the sleeping \congest{} model that computes an
$\tilde{O}(\alpha\log_\alpha\Delta)$-approximate dominating set in expectation in $\tilde{O}(\log\Delta\cdot \log_{\alpha} \Delta)$ rounds with $\tilde{O}(\log_{\alpha} \Delta)$ awake complexity.  
Our results depend on a generalization of the \congest{} model \textsc{SetCover} algorithm of Grunau, Mitrovi\'c, Rubinfeld, and Vakilian (SODA 2020) that we develop. 
This generalization computes an $O(\stagebase\cdot\iterbase\cdot\log_\stagebase\Delta)$-approximate dominating set in $O(\log_\stagebase\Delta\cdot\log_\iterbase\Delta)$ rounds for parameters $1<\stagebase,\iterbase\le\Delta$. 
Our sleeping \congest{} algorithms apply a variety of techniques including sampling-based estimation and scheduling using virtual binary trees to the aforementioned 2-parameter \textsc{SetCover} algorithm.
\end{abstract}

\noindent\textbf{CCS Concepts:} Theory of computation~Distributed algorithms; Theory of computation~Approximation algorithms analysis

\noindent\textbf{Keywords:} minimum dominating set, sleeping CONGEST, awake complexity, distributed approximation

\section{Introduction}

Given an undirected graph $G=(V,E)$, a \emph{dominating set} of $G$ is a subset $D\subseteq V$ such that every vertex in $V\setminus D$ has a neighbor in $D$.  
The \emph{Minimum Dominating Set (MDS)} problem is to find a dominating set of smallest size.
MDS is a classical NP-complete combinatorial optimization problem and the natural, greedy sequential algorithm provides a $(1+\ln\Delta)$-approximation~\cite{johnson1973approximation,chvatal1979greedy}. 
This algorithm is essentially optimal since, via the standard reduction from \textsc{SetCover} to MDS, there can be no $c\ln n$-approximation to MDS for any constant $c<1$ unless $P=NP$~\cite{dinur2014analytical}. 
In the distributed setting, MDS has been studied extensively in the standard \congest{} and \local{} models. 
In the \congest{} model, Jia, Rajaraman, and Suel~\cite{jia2002efficient} gave a randomized algorithm running in $O(\log n\log \Delta)$ rounds with high probability, providing an $O(\log\Delta)$-approximation (in expectation) and an $O(\log n)$-approximation (with high probability). 
Censor-Hillel and Dory~\cite{censor2018distributed} subsequently strengthened this result to a high-probability $O(\log\Delta)$-approximation in $O(\log n\log\Delta)$ rounds in the \congest{} model. 
Kuhn, Moscibroda, and Wattenhofer~\cite{kuhn2006price} present a similar result via solving covering-packing linear programs (LPs).
As far as we know, for general graph, no algorithm achieving an $O(\log\Delta)$-approximation in $o(\log^2\Delta)$ rounds is currently known in the \congest{} model. 
However, in the \local{} model (i.e., using messages of unbounded size), Kuhn, Moscibroda, and Wattenhofer~\cite{kuhn2006price,kuhn2016local} show how to compute an expected $O(\log \Delta)$-approximation in $O(\log n)$ rounds, with high probability.

The \local{} and the \congest{} models charge every node for every round of execution, whether or not the node is actively communicating in that round. 
In many real-world distributed systems---sensor networks, ad-hoc wireless networks, and other battery-powered settings---energy, not time, is the scarce resource: an idle but awake node consumes nearly as much energy as one that is transmitting or receiving, while a sleeping node (with its radio off) consumes essentially none. 
This observation motivated a substantial line of work on energy-efficient radio protocols~\cite{chang2017leader,chang2018broadcast}, and was carried over to the message-passing distributed-computing setting by Chatterjee, Gmyr, and Pandurangan~\cite{chatterjee2020sleeping}, who introduced the \emph{sleeping} model: a generalization of the standard \congest{} and \local{} models in which a node may sleep in any round, and only the rounds in which it is awake count toward its individual \emph{awake complexity} (also called \emph{energy complexity}). 
Crucially, the awake complexity of an algorithm can be substantially smaller than its round complexity. 

A growing body of recent work designs algorithms in the sleeping \congest{} and \local{} models with low awake complexity for graph problems that are related to MDS.
For the \textit{maximal independent set} (MIS) problem in the sleeping \congest{} model, Chatterjee, Gmyr, and Pandurangan~\cite{chatterjee2020sleeping} gave a randomized algorithm that computes an MIS with high probability, has expected $O(1)$ node-averaged awake complexity, and with high probability, has $O(\log n)$ worst-case awake complexity and $O(\log^{3.41} n)$ round complexity.
Dufoulon, Moses~Jr., and Pandurangan~\cite{dufoulon2023mis} subsequently improved the worst-case awake complexity to $O(\log\log n)$ rounds, though the algorithm uses $O(\log^7 n\log\log n)$ rounds.
Ghaffari and Portmann~\cite{ghaffari2023lowenergy} reduced the round complexity to $O(\log^2 n)$, while maintaining the worst-case awake complexity at $O(\log\log n)$ rounds.
For maximum matching and minimum vertex cover, Ghaffari and Portmann~\cite{ghaffari2022average} presented $O(\log n)$-round algorithms in the sleeping \congest{} model with $O(1)$ node-averaged awake complexity for $(1+\varepsilon)$-approximate maximum matching and $(2+\varepsilon)$-approximate minimum vertex cover, for every constant $\varepsilon>0$. 

O-LOCAL problems refer to the class of problems that can be solved using any acyclic orientation of the edges, by choosing a solution for each vertex after all vertices reachable from it have computed their solution, and as a function of these solutions.
Problems such as MIS and $(\Delta+1)$-coloring are O-LOCAL problems.
For this class of problems, Barenboim and Maimon~\cite{barenboim2021deterministic} gave a deterministic algorithm using $O(\log\Delta+\log^* n)$ awake rounds and $O(\Delta^2+\log^* n)$ rounds in the sleeping \local{} model.
Balliu et al.~\cite{balliu2025sequential} subsequently showed that every O-LOCAL problem can be solved deterministically in the sleeping \local{} model using $O(\sqrt{\log n}\log^* n)$ awake rounds and $O(n^5\sqrt{\log n})$ rounds.
MDS is not O-LOCAL and so this general result does not have any direct implications for MDS.

For MDS, prior sleeping-model work has been confined to special graph classes arising in wireless networks~\cite{lam2025minimum} (see Section \ref{subsec:independentResults} for an update). 
As mentioned earlier, the state of the art for MDS in the \congest{} model is an $O(\log \Delta)$-approximation algorithm running in $O(\log^2 \Delta)$ rounds.
No faster algorithms are known to provide this approximation guarantee.
Given this situation, our work is motivated by the following question.
\begin{mdframed}[linewidth=1pt,linecolor=black,backgroundcolor=gray!10,roundcorner=5pt]
    \textbf{Core Question:} Can we design an $O(\log \Delta)$-approximation for MDS with $o(\log^2 \Delta)$ awake complexity?
\end{mdframed}

A natural generalization of this question is whether it is possible to design algorithms that exhibit a tradeoff between approximation quality and awake complexity, specifically allowing the approximation guarantee to worsen in exchange for further improvement in the awake complexity.
Our work in this paper is motivated by this ``approximation versus energy'' tradeoff question as well.

\subsection{The Sleeping Model}
\label{subsec:sleepingCongest}

We work in the synchronous sleeping model of Chatterjee, Gmyr, and Pandurangan~\cite{chatterjee2020sleeping}, which generalizes the standard \congest{} model by allowing each node to sleep in any round. 
The algorithm proceeds in synchronous rounds, and in each round every node is either awake or asleep. An awake node performs local computation and may send messages, of size $O(\log n)$ bits, to each of its neighbors. 
A sleeping node performs no computation, sends no messages, and \textit{receives no messages}. Thus any message sent to a sleeping node is lost. A sleeping node wakes only at a round it previously scheduled for its wake up; it cannot be woken by a neighbor. 

We let $\Delta$ denote the maximum closed-neighborhood size (equivalently, $1$ plus the maximum degree) and assume that $\Delta \ge 2$ (otherwise the MDS problem becomes trivial).
We assume that each node initially knows $n$ and $\Delta$, and that on waking up a node knows the current round number, presumably by accessing the global clock. 
A node $v$ that is awake in a round $i$, schedules its next awake round $j > i$, before completing round $i$. 
Node $v$ then sleeps during rounds $i+1, i+2, \ldots, j-1$ and wakes up in round $j$ with no help from any other nodes. 
The \emph{round complexity} of an algorithm is the number of rounds until every node terminates. The \emph{awake complexity} of an algorithm is $\max_{v\in V} A_v$, where $A_v$ is the number of rounds in which $v$ is awake.

\subsection{Our results}

We present two main results on MDS approximation in the sleeping \congest{} model. Our first result affirmatively answers our core question. 

\begin{theorem}\label{thm:homestage-mds} 
There is an algorithm in the sleeping \congest{} model that computes an $O(\log\Delta)$-approximate dominating set in expectation in $O(\log^2\Delta)$ rounds, with each node awake in at most $O(\log \Delta)$ rounds. 
\end{theorem}

We obtain this result by designing a sleeping version of the \textsc{SetCover} algorithm of Grunau, Mitrovi\'c, Rubinfeld, and Vakilian~\cite{grunau2020improved}. 
For MDS, this algorithm yields an $O(\log \Delta)$-approximation (in expectation) in $O(\log^2 \Delta)$ rounds. 
We note that even though in our sleeping version, nodes spend large portions of the algorithm asleep, this does not result in the slowing down of the algorithm, which still runs in $O(\log^2 \Delta)$ rounds.
We develop the sleeping implementation in two steps. Section~\ref{sec:base-awake} first gives a basic implementation with $O(\log\Delta\log\log n)$ awake complexity. Section~\ref{subsec:refined-cutoff} then refines the analysis of estimation errors and uses the generalized algorithm of Section~\ref{sec:basemds-qary} to accelerate the cleanup phase, obtaining the $O(\log\Delta)$ awake bound. Theorem~\ref{thm:base-awake-refined} gives the formal statement.

Our second result establishes a smooth tradeoff between the approximation factor and awake complexity of MDS algorithms, specifically showing there are MDS approximation algorithms with $o(\log \Delta)$ awake complexity, provided we permit weaker approximation guarantees.

\begin{theorem}[Informal]\label{thm:homestage-mds-tradeoff}
For any $1<\alpha\le\min\{\Delta,\log n\}$, there is an algorithm in the sleeping \congest{} model that computes an $\tilde{O}(\alpha\log_\alpha\Delta)$-approximate dominating set in expectation in $\tilde{O}(\log\Delta\cdot \log_{\alpha} \Delta)$ rounds and $\tilde{O}(\log_{\alpha} \Delta)$ awake complexity.\footnote{We use the $\tilde{O}(\cdot)$ notation in approximation, round complexity, and awake complexity bounds to hide $O(\mbox{poly}(\log\log n))$ factors.}
\end{theorem}

This result (see Theorem \ref{thm:homestage-mds-tradeoff-formal} for the precise version of this theorem) allows us to obtain sublogarithmic awake complexities, such as $\tilde{O}(1)$, for $\alpha=\Delta^\eps$ when $\Delta^\eps\le\log n$, and $\tilde{O}(\log^{1-\eps}\Delta)$, for $\alpha=2^{\log^\eps\Delta}$ when $2^{\log^\eps\Delta}\le\log n$, for constant $\eps\in(0,1)$.
We use a variety of techniques to obtain our results---see Section \ref{subsec:challenges} below for an overview. 
One new result that we obtain, that may be of independent interest, is the following natural, 2-parameter \congest{} model generalization of the \textsc{SetCover} algorithm of Grunau, Mitrovi\'c, Rubinfeld, and Vakilian~\cite{grunau2020improved}.

\begin{theorem}\label{thm:pq-basemds}
For every $1<\stagebase,\iterbase\le\Delta$, there is a randomized \congest{} algorithm that returns an $O(\stagebase\cdot \iterbase\cdot \log_\stagebase\Delta)$-approximate dominating set in expectation, in $O(\log_\stagebase\Delta\cdot\log_\iterbase\Delta)$ rounds.
\end{theorem}

By setting $\stagebase=\iterbase=2$ we recover the $O(\log\Delta)$-approximation in $O(\log^2 \Delta)$ rounds due to Grunau, Mitrovi\'c, Rubinfeld, and Vakilian~\cite{grunau2020improved}.
We use Theorem \ref{thm:pq-basemds}, with different settings of $\stagebase$ and $\iterbase$ to obtain both Theorems \ref{thm:homestage-mds} and \ref{thm:homestage-mds-tradeoff}.

\subsection{Independent results}
\label{subsec:independentResults}

While our paper was being reviewed, the paper of Ghaffari and Yeoh~\cite{ghaffari2026energy} was published in SPAA 2026. 
Among other results, this paper gives a $\widetilde O(\log^2 n)$-approximation \congest{} algorithm with high probability for minimum set cover in $O(\log^2 n)$ rounds, with $O(\log\log n)$ awake complexity for element nodes and $O(\log n\log\log n)$ awake complexity for set nodes.
Since MDS is a special case of set cover, their result also implies an $\widetilde O(\log^2 n)$-approximate MDS algorithm with $O(\log n\log\log n)$ awake complexity.
Interestingly, the starting point of the Ghaffari and Yeoh~\cite{ghaffari2026energy} energy-efficient algorithm is the same, \textsc{BaseMDS} algorithm (Algorithm \ref{alg:base-mds}) \cite{grunau2020improved} that we use.
The paper also presents improved results on the worst case awake complexity for maximum matching and minimum vertex cover approximation.

\subsection{Challenges and ideas}
\label{subsec:challenges}

\textbf{The \textsc{BaseMDS} algorithm.}
Our algorithmic starting point is a parallel emulation of the classical sequential greedy MDS algorithm, which repeatedly picks the vertex $v$ whose closed neighborhood $N^{+}(v)$ ($= \{v\} \cup N(v)$) contains the largest number of still-undominated vertices and adds $v$ to the dominating set $D$.  
This achieves a $(1+\ln\Delta)$-approximation~\cite{johnson1973approximation,chvatal1979greedy}.  A direct distributed parallelization is hopeless: the algorithm seems inherently sequential and identifying the global maximum-residual-degree vertex would by itself cost $\Omega(\operatorname{diam}(G))$ rounds.  
Here ``residual degree'' of a node $v$ refers to the number of as-yet-undominated neighbors of node $v$ (including $v$ if it is not dominated). 
The standard relaxation, going back to early parallel \textsc{SetCover} algorithms~\cite{berger1994nc,rajagopalan1998primal} and also used in the algorithm of Grunau, Mitrovi\'c, Rubinfeld, and Vakilian~\cite{grunau2020improved} that we build on, is to replace ``the vertex of maximum residual degree'' by ``all vertices whose residual degree is high enough.''  
Concretely, we identify geometrically decreasing residual degree thresholds $T_i=\Delta/2^i$, $1 \le i \le \lceil \log \Delta \rceil$ and let the algorithm proceed in $\lceil\log\Delta\rceil$ \emph{stages}.
At the start of stage $i$, any vertex whose residual degree is at least $T_i$ is considered a \textit{candidate} for Stage $i$.

All stage $i$ candidates would like to join $D$ in stage $i$, but a stage may have many candidates whose closed neighborhoods overlap, and adding all of them simultaneously could lead to a very poor approximation.
So we need to run some type of \textit{symmetry breaking} algorithm in each stage to prune the candidates in that stage.
The symmetry breaking process needs to ensure that the selected Stage $i$ candidates satisfy two properties: 
(i) every as-yet-undominated node covered by some Stage $i$ candidate is covered by some \textit{selected} stage $i$ candidate and 
(ii) no as-yet-undominated node is ``over covered'' by the selected stage $i$ candidates. 
Several approaches have been proposed for symmetry breaking in this context (e.g., randomized pruning based on median support \cite{jia2002efficient,PanditPemmaraju2010}, random voting based scheme in \cite{censor2018distributed}). 
The symmetry breaking used in the algorithm of Grunau, Mitrovi\'c, Rubinfeld, and Vakilian~\cite{grunau2020improved} is particularly simple.
In each stage, we simply run $\lceil\log\Delta\rceil$ \emph{iterations} and in iteration $k$, every still-eligible candidate joins $D$ independently with probability $\probit_k=\min\{1,2^k/\Delta\}$.  
The fact that the probabilities start off very small, at $2/\Delta$, and grow geometrically in each iteration ensures that each as-yet-undominated node is dominated by only $O(1)$ selected candidates in expectation.  
This is a key result in the analysis of~\cite{grunau2020improved}. 
We refer to this algorithm from \cite{grunau2020improved} as \textsc{BaseMDS}. Algorithm~\ref{alg:base-mds} presents the pseudocode.

\begin{algorithm}[t]
\caption{\textsc{BaseMDS}}
\label{alg:base-mds}
\begin{algorithmic}[1]
\Require maximum closed-neighborhood size $\Delta$ known to all nodes.
\State $D\gets\emptyset$.
\For{$i=1,2,\dots,\lceil\log \Delta\rceil$} \Comment{stages}
\For{$k=1,2,\dots,\lceil\log \Delta\rceil$} \Comment{iterations}
    \For{each node $v$ in parallel}
      \State $d_i^k(v)\gets|\{u\in N^{+}(v):u\notin N^{+}(D)\}|$.
      \If{$v\notin D$ and $d_i^k(v)\ge T_i$} \Comment{$T_i = \Delta/2^i$}
        \State add $v$ to $D$ with probability $\probit_k = \min\{1,2^k/\Delta\}$.
      \EndIf
    \EndFor
  \EndFor
\EndFor
\State \Return $D$.
\end{algorithmic}
\end{algorithm}

The \textsc{BaseMDS} algorithm computes an $O(\log\Delta)$-approximate dominating set in expectation in $O(\log^2\Delta)$ \congest{} rounds~\cite{grunau2020improved}. 
We prove our main results described earlier by designing an energy efficient version of \textsc{BaseMDS} (Algorithm~\ref{alg:base-mds}) (and its generalization Algorithm~\ref{alg:pq-mds}). In doing this, we encountered some significant challenges that we discuss below.

\medskip
\noindent
\textbf{Challenge 1: Nodes continually need to monitor residual degrees.}
The most basic challenge in designing a sleeping version of \textsc{BaseMDS} is that nodes need to continually monitor their dynamically decreasing residual degrees so that they can attempt to join $D$ when their residual degree exceeds a threshold $T_i$.
While the thresholds $T_i$ are fixed ahead of time, it is hard to anticipate how the residual degrees will fall from one iteration to the next.
Specifically, each node $v$ in \textsc{BaseMDS} needs to monitor its residual degree at the start of each stage to determine if it is a candidate for that stage.
If a node $v$ is a stage $i$ candidate, it needs to continually monitor its residual degree during the iterations in stage $i$ in order to participate in symmetry breaking.
It is not difficult to construct worst-case examples in which a node becomes a candidate in most stages and therefore has to monitor its degree for $\Theta(\log^2 \Delta)$ iterations over the entire algorithm.
Our first conceptual contribution is showing that \textit{if a node $v$ is a candidate in stage $i$, it only needs to wake up in $1$ iteration during the stage to participate in symmetry breaking}. 
Furthermore, it can identify its wake-up iteration at the start of the stage in order to schedule its wake-up. 

\medskip
\noindent
\textbf{Challenge 2: Nodes continually need to share their status with neighbors.}
Overcoming Challenge 1 immediately implies that each node needs to be awake as a candidate for $O(\log \Delta)$ rounds, i.e., $O(1)$ rounds per stage. 
However, whenever a node wakes up to monitor its residual degree, its neighbors also need to be awake in order to provide up-to-date information on their status (dominated or not). This means that nodes may have to be continually awake to respond to queries from neighbors about their status.
Our second conceptual contribution deals with this challenge. To each node $u$, we assign a single randomly chosen stage-iteration pair over the entire algorithm, and during this pair, node $u$ acts as an \textit{estimator}, providing status updates to any neighbor querying it. 
This essentially means that for any node $v$ that needs to compute its residual degree in an iteration, the estimators in $v$'s neighborhood that are awake during that iteration form a random neighborhood sample. 
In the basic implementation, we show in Lemma~\ref{lem:est-decision-base} that as long as the threshold $T_i = \Omega(\text{poly}(\log n))$, this energy-saving sampling scheme is enough for node $v$ to obtain an accurate estimate of its residual degree.

\medskip
\noindent
\textbf{Challenge 3: Coordinating dominators and estimators.}
Our next challenge arises from the fact that nodes that wake up in just one iteration, to act as estimators, will have no clue about their status (dominated or not). Specifically, they would not have heard from any neighbors who joined $D$ in a prior iteration because they were not awake at that time.
So we need to somehow coordinate the wake up schedule of a node $w$ that wakes up as an estimator in iteration $j$ and the wake up schedule of a neighbor $w' \in N(w)$ that joins $D$ in an earlier iteration $i < j$.
Schedules obtained by using a \textit{virtual binary tree (VBT)} \cite{ghaffari2023lowenergy,barenboim2021deterministic,augustine2024mst} have been proposed to solve precisely this coordination challenge.
Continuing our example, let $\mathcal{S} = \{1, 2, \ldots, \lceil \log \Delta \rceil\}$ denote the iteration numbers in a stage of \textsc{BaseMDS}.
The node $w$, that is scheduled to wake up in iteration $j$ uses a VBT to compute a set $S_j \subseteq \mathcal{S}$ of iterations and similarly the node $w' \in N(w)$ that joins $D$ in iteration $i$ uses a VBT to compute a set $S_i \subseteq \mathcal{S}$ of iterations.
The VBT-based schedule has the property that (i) the sets are small, i.e., $|S_i|, |S_j| = O(\log |\mathcal{S}|)$ and (ii) there is an iteration $\ell \in S_i \cap S_j$ such that $i \le \ell \le j$.
See Lemma~\ref{lem:vbt-base} for a precise statement of the VBT schedule property.
This implies if $w$ is awake for all iterations in $\{j\} \cup S_j$ and $w'$ is awake for all iterations in $\{i\} \cup S_i$ then there is a common iteration $\ell$, $i \le \ell \le j$, during which both $w$ and $w'$ are awake and during this iteration $w$ can learn about $w'$ and can update its status so that it can correctly answer status queries in iteration $j$.
The fact that the sizes $|S_i|, |S_j| \in O(\log \log \Delta)$ implies that all of this can be achieved with only an $O(\log\log \Delta)$ overhead in the awake complexity of nodes.

\medskip
\noindent
\textbf{Challenge 4: Reducing awake complexity to $o(\log \Delta)$.}
Given that each node wakes up a constant number of times per stage to monitor its degree and there are $\Theta(\log \Delta)$ stages, how do we reduce the awake complexity of the algorithm to $o(\log \Delta)$?
Our approach to doing this is by reducing the number of stages in the algorithm.
Specifically, we generalize the \textsc{BaseMDS} algorithm by allowing the degree thresholds to fall by a factor $\stagebase > 1$ in each stage, rather than by a factor 2.
Using a parameterized degree threshold sequence $T_i = \Delta/\stagebase^i$ implies $\Theta(\log_\stagebase \Delta)$ stages and using values of $\stagebase > 2$ reduces the number of stages, while incurring a corresponding loss in approximation factor.
In fact, we also parameterize the sampling probability sequence that define iterations, using $\min\{1, \frac{\iterbase^k}{\Delta}\}$, instead of $\min\{1, \frac{2^k}{\Delta}\}$, for a parameter $\iterbase > 1$. 
The resulting algorithm is a natural 2-parameter \congest{} model generalization of \textsc{BaseMDS} that we call $(p, q)$-MDS. 
We show, by generalizing the analysis of Grunau, Mitrovi\'c, Rubinfeld, and Vakilian~\cite{grunau2020improved}, that this algorithm yields a $O(\stagebase \cdot \iterbase \cdot \log_\stagebase \Delta)$-approximation to MDS (in expectation), while running in $O(\log_\stagebase \Delta \cdot \log_\iterbase \Delta)$ rounds.
The sleeping implementation of this generalization gives the tradeoff in
Theorem~\ref{thm:homestage-mds-tradeoff}.
The generalization also lets us accelerate the cleanup phase to obtain Theorem~\ref{thm:homestage-mds}. In Section~\ref{subsec:refined-cutoff}, we account for incorrect estimator decisions in expectation, allowing the cutoff between the two phases to depend only on $\Delta$. This leaves $O(\log\log\Delta)$ cleanup stages; choosing a larger $\iterbase$ for these stages reduces their total cost to $O(\log\Delta)$ awake rounds while preserving the expected $O(\log\Delta)$ approximation.

There have been several papers \cite{kuhn2003constant,kuhn2006price,PanditPemmaraju2010} on approximation versus rounds tradeoff for MDS in the \congest{} model.
For example, Kuhn and Wattenhofer~\cite{kuhn2003constant} gave an $O(k\Delta^{2/k}\log\Delta)$-approximation in $O(k^2)$ rounds for each integer $k\ge 1$.
Our $(p,q)$-MDS algorithm gives a tradeoff of the same general form: setting $p=q=\Delta^{1/k}$ in Theorem~\ref{thm:pq-basemds} gives an $O(k\Delta^{2/k})$-approximation in expectation in $O(k^2)$ rounds. We do not use $(p,q)$-MDS to improve the best known round--approximation tradeoff. Rather, its staged threshold structure is particularly amenable to our sleeping implementation, allowing the same parameters to control the awake complexity as well.

\subsection{Other related work}

Global problems such as \textit{Minimum Spanning Tree} (MST) and \textit{Single Source Shortest Path} (SSSP) have also been studied in distributed settings from the energy perspective.
For MST in the sleeping \congest{} model, Augustine, Moses~Jr., and Pandurangan~\cite{augustine2024mst} proved that the optimal awake complexity is $\Theta(\log n)$. They gave a randomized algorithm that, with high probability, computes an MST using $O(\log n)$ awake rounds and $O(n\log n)$ rounds, as well as a deterministic algorithm using $O(\log n)$ awake rounds and $O(n\log^5 n)$ rounds.
For SSSP in the sleeping \congest{} model, Ghaffari and Trygub~\cite{GhaffariTrygubPODC2024} present a deterministic algorithm that runs in $\widetilde{O}(n)$ rounds and each node is awake during only $\text{poly}(\log n)$ rounds.

\section{Awake-efficient \textsc{BaseMDS}}
\label{sec:base-awake}

We first present a basic sleeping \congest{} version of \textsc{BaseMDS} (Algorithm~\ref{alg:base-mds})~\cite{grunau2020improved} that computes an $O(\log\Delta)$-approximate dominating set in expectation in $O(\log^2\Delta)$ rounds, with each node awake in $O(\log\Delta\log\log n)$ rounds. 
The stated round and awake-complexity bounds hold deterministically.
This implementation introduces the estimation and wake schedules used throughout the paper. In Section~\ref{subsec:refined-cutoff}, we refine its analysis and cleanup phase to obtain the $O(\log\Delta)$ awake complexity stated in Theorem~\ref{thm:homestage-mds}.
We use the same thresholds $T_i=\Delta/2^i$ and probabilities $\probit_k=\min\{1,2^k/\Delta\}$ as in \textsc{BaseMDS}. 
Also recall that $d_i^k(v)$ denotes the number of undominated nodes in the closed neighborhood of $v$ just before iteration $k$ of stage $i$, with respect to the current dominating set $D$, i.e., 
$d_i^k(v) =  |\{u\in N^{+}(v):u\notin N^{+}(D)\}|$. 
We will refer to $d_i^k(v)$ as the \textit{residual degree} of $v$ during iteration $k$ of stage $i$.
We note that the algorithm only uses $1$-bit messages.

The algorithm has two phases. Let $i^\star = \min\{i:T_i\le C\lceil\log\Delta\rceil^2\log n\}$, where $C$ is a constant that will be fixed later (in Lemma~\ref{lem:est-decision-base}). 
Phase~1 runs stages $1,\dots,i^\star-1$. During Phase~1 most nodes sleep most of the time, and a candidate decides whether to join $D$ from a sampling-based estimate of its residual degree. 
Phase~2 runs the remaining $O(\log\log n)$ stages exactly as in \textsc{BaseMDS}, with every node awake in every round. We split the algorithm in this manner because we need sufficiently many samples for a node to obtain an accurate estimate of its residual degree (Lemma~\ref{lem:est-decision-base}).
Before Phase~1 starts, each node $u$ picks an estimator assignment $(I_u,K_u)\in [\lceil\log\Delta\rceil] \times [\lceil\log\Delta\rceil]$ uniformly at random. If $I_u < i^\star$, node $u$ wakes up in iteration $K_u$ of stage $I_u$ and acts as an \emph{estimator}, replying to queries from its neighbors.
Thus each node acts as an estimator in at most one iteration during Phase~1.

We now describe Phase~1 of our algorithm.
Each stage $i$, $1 \le i < i^\star$, begins with two status-exchange rounds in which every node is awake. In the first of these rounds, every $v\in D$ broadcasts this fact, and any node hearing such a broadcast from a neighbor records that it is dominated. 
In the second round, every undominated node broadcasts this status. Using this information, each node $v$ then computes $d_i^1(v)$, its residual degree during iteration 1 of stage $i$. 
A node $v$ with $d_i^1(v)\ge T_i$ becomes a \emph{candidate} for stage $i$.
Each candidate $v$ then tosses a sequence of 
$\lceil\log\Delta\rceil$ biased coins, with coin $k$ in the sequence having success probability $\probit_k$.
Let $J_{v,i} \in [\lceil\log\Delta\rceil]$ denote the minimum index $k$ such that the outcome of coin $k$ is a success.
Since $p_{\lceil\log\Delta\rceil}$ is 1, $J_{v,i}$ is well-defined.
As a candidate, $v$ only needs to wake in iteration $J_{v,i}$ to decide whether to join $D$. 

The stage is then run for $\lceil\log\Delta\rceil$ iterations, each with two rounds. 
In round 1 of iteration $k$, every awake node $w\in D$ that joined $D$ earlier in stage $i$ informs all neighbors of this fact and any awake node hearing from a neighbor in $D$ marks itself dominated. 
In round 2 of iteration $k$, every node $u$ with $(I_u,K_u)=(i,k)$ that was undominated at the start of stage $i$ sends its status (i.e., dominated or undominated) to all neighbors. 
Note that such a node $u$ has picked iteration $k$ in stage $i$ to serve as an estimator. 
Let $X$ be the number of nodes in $N^{+}(v)$ that were undominated at the start of stage $i$ and replied to $v$ in this round (including $v$ itself if $(I_v,K_v)=(i,J_{v,i})$ and $v$ was undominated at the start of stage $i$), and let $Y$ be the number of these that are still undominated. The candidate $v$ then computes
\[
   \widehat d(v)=\frac{Y}{X}\,d_i^1(v)\quad\text{if }X>0,\qquad\widehat d(v)=0\quad\text{if }X=0.
\]
as an estimator of its residual degree in stage $i$, iteration $k$.
If $\widehat d(v)\ge T_i$, then $v$ joins $D$.

The above description ignores Challenge 3 from the ``Introduction'' which asks how nodes that wake up in iteration $k$ in stage $i$ as estimators, can know their status accurately. 
The next subsection describes the virtual binary tree idea \cite{barenboim2021deterministic,augustine2024mst,ghaffari2023lowenergy} as it applies to our algorithm. This results in each node staying awake for an additional $O(\log \log \Delta)$ rounds, so that estimators know their status accurately.

\subsection{Wake schedule via virtual binary trees}
\label{subsec:base-vbt}

We want to ensure that candidates that join $D$ in stage $i$ inform neighbors that wake up later as estimators in stage $i$.
We use the \textit{virtual binary tree} idea described in previous work \cite{barenboim2021deterministic,augustine2024mst,ghaffari2023lowenergy}, specifically the formulation of Ghaffari and Portmann~\cite{ghaffari2023lowenergy}, restated below.

\begin{lemma}[\cite{ghaffari2023lowenergy}, Lemma~2.5]\label{lem:vbt-base}
For any positive integer $T$, there exist sets $S_1,\dots,S_T\subseteq[T]$ with $|S_k|\le\lceil\log T\rceil$ such that for every $i,j\in[T]$ with $i\le j$ there exists $\ell\in S_i\cap S_j$ satisfying $i\le\ell\le j$.
\end{lemma}

We use this lemma as follows. Consider two neighboring nodes $u$ and $v$, scheduled to wake up in iterations $k_u$ and $k_v$, $k_u \le k_v$, in stage $i$.
Now suppose that instead of waking up just in iterations $k_u$ and $k_v$, node $u$ wakes up in all iterations $k \in \{k_u\} \cup S_{k_u}$ in stage $i$ and node $v$ wakes up in all iterations $k \in \{k_v\} \cup S_{k_v}$ in stage $i$. Then the lemma guarantees that there exists an iteration $\ell \in S_{k_u} \cap S_{k_v}$ in stage $i$ in which both $u$ and $v$ are awake and they can share information in that common iteration.

Replacing each $S_k$ with $S_k\cup\{k\}$, we may assume $k\in S_k$ for every $k$. 
Since each stage has $T=\lceil\log\Delta\rceil$ iterations, applying 
Lemma~\ref{lem:vbt-base}, we see that every set $S_k$ satisfies $|S_k| = O(\log\log \Delta)$, which in turn implies that every node that is scheduled to wake up in stage $i$ will wake up in $O(\log\log \Delta)$ iterations.
Note that all nodes can compute all sets $\{S_k\}_{k=1}^{T}$ at the start of the algorithm and reuse them in every stage for determining at the start of the stage, which iterations to wake up in.
We now provide more details about the wake up schedules of estimators and candidates separately.

\emph{Estimator schedule.} An estimator $u$ with $(I_u,K_u)=(i,k)$ wakes at every $\ell\in S_k$ with $\ell\le k$. Node $u$ acts as an estimator in iteration $k$. The earlier wake-up iterations allow $u$ to receive updates from any neighbor that joined $D$ in an earlier iteration of stage $i$.

\emph{Candidate schedule.} A candidate $v$ with $J_{v,i}=k$ wakes up in iteration $k$ to determine its residual degree and then possibly join $D$. If $v$ joins $D$ at iteration $k$, it wakes in every later iteration $\ell\in S_k$ with $\ell>k$ to inform all neighbors (who are awake in that iteration) of this fact.
Otherwise (if $v$ does not join $D$ in iteration $k$), $v$ sleeps for the remainder of the stage.

By Lemma~\ref{lem:vbt-base}, for any $w$ that joined $D$ in stage $i$ and any estimator $u$ with $J_{w,i}\le K_u$, the sets $S_{J_{w,i}}$ and $S_{K_u}$ intersect at some $\ell\in[J_{w,i},K_u]$. 
Thus, at iteration $\ell$, node $w$ informs node $u$ that it has joined $D$.
Consequently, before acting as an estimator in iteration $K_u$, node $u$ knows that it has become dominated. Algorithm~\ref{alg:base-awake} presents the full pseudocode.

\begin{algorithm}[t]
\caption{\textsc{BaseMDS-Awake}}
\label{alg:base-awake}
\begin{algorithmic}[1]
\Require $\Delta$ known to all nodes; constant $C$ from Lemma~\ref{lem:est-decision-base}; wake schedules $\{S_k\}_{k=1}^{\lceil\log\Delta\rceil}$ from Lemma~\ref{lem:vbt-base}.
\vspace{1ex}
\State $D\gets\emptyset$;\, $i^\star\gets\min\{i:T_i \le C\lceil\log\Delta\rceil^2\log n\}$\, \Comment{$T_i=\Delta/2^i$}
\State Each $v$ draws $(I_v,K_v)\sim\mathrm{Unif}([\lceil\log\Delta\rceil] \times [\lceil\log\Delta\rceil])$.
\vspace{1ex}
\Statex \textbf{Phase~1: sampling-based estimation.}
\For{$i=1,\dots,i^\star-1$} \Comment{stages}
  \State \parbox[t]{\linewidth-\algorithmicindent-\algorithmicindent}{\textbf{Round 1}: every awake candidate in $D$ informs neighbors that it is in $D$; any awake node hearing this marks itself dominated.}
  \State \parbox[t]{\linewidth-\algorithmicindent}{\textbf{Round 2}: every node wakes; every undominated node broadcasts this status, and each $v$ uses this information to compute $d_i^1(v)=|N^{+}(v)\setminus N^{+}(D)|$.}
  \State \parbox[t]{\linewidth-\algorithmicindent}{Each $v$ for which $d_i^1(v)\ge T_i$ is a candidate.} 
  \State \parbox[t]{\linewidth-\algorithmicindent}{Each candidate $v$ computes first-success iteration $J_{v,i}\in[\lceil\log\Delta\rceil]$ by tossing a sequence of biased coins with $\probit_k=\min\{1,2^k/\Delta\}$.
  \Comment{Lemma~\ref{lem:first-success-sampling-base}}}
  \For{$k=1,\dots,\lceil\log\Delta\rceil$} \Comment{iterations}
    \State Every candidate $v$ with $J_{v,i}=k$ wakes up.
    \State Every candidate $w\in D$ that joined earlier in stage $i$ wakes up
    if $k\in S_{J_{w,i}}$.
    \State Every estimator $u$ with $I_u = i$ wakes up if $k \in S_{K_u}$ and $k \le K_u$.
    \State \parbox[t]{\linewidth-\algorithmicindent-\algorithmicindent}{\textbf{Round 1}: every awake candidate informs neighbors that it is in $D$; any awake node hearing this marks itself dominated.}
    \State \parbox[t]{\linewidth-\algorithmicindent-\algorithmicindent}{\textbf{Round 2}: every $u$ with $(I_u,K_u)=(i,k)$ that was undominated at the start of stage $i$ informs all neighbors of its status.}
    \State \parbox[t]{\linewidth-\algorithmicindent-\algorithmicindent}{Each candidate $v$ computes its estimate $\widehat d(v)$ from the received statuses; if $\widehat d(v)\ge T_i$, $v$ joins $D$, and schedules wake-ups at $\{\ell\in S_k:\ell>k\}$.}
  \EndFor
\EndFor
\Statex \textbf{Phase~2: exact cleanup.}
\State Run stages $i^\star,\dots,\lceil\log\Delta\rceil$ of \textsc{BaseMDS} (Algorithm~\ref{alg:base-mds}) exactly on the residual instance, with every node awake.
\State \Return $D$.
\end{algorithmic}
\end{algorithm}

\subsection{Analysis}
\label{subsec:base-awake-analysis}

In transforming \textsc{BaseMDS} into an energy-efficient algorithm (Algorithm \ref{alg:base-awake}), we ensured that each node wakes up only $O(1)$ times per stage by virtue of it being a candidate. We first show that this transformation is correct.

Consider a node $v$ and suppose that $v$ is a candidate for stage $i$. In other words, $d_i^1(v) \ge T_i$.
Consider the following two executions for node $v$ during the iterations in stage $i$.

\begin{description}
\item[Execution 1:] For each iteration $k = 1, 2, \ldots, \lceil \log \Delta\rceil$, node $v$ tosses a coin with probability $\probit_k = \min\{1, 2^k/\Delta\}$ and joins $D$ at the first iteration $k$ during which $d_i^k(v) \ge T_i$ and the coin $k$ is successful. 
\item[Execution 2:] Set $J=\min\{k:\text{coin }k\text{ is successful}\}$. Node $v$ joins $D$ at iteration $J$ iff  
$d_i^J(v) \ge T_i$.
\end{description}
The following lemma shows that the two executions are equivalent.

\begin{lemma}[First-success sampling]
\label{lem:first-success-sampling-base}
If the coins used in executions 1 and 2 are coupled, then either node $v$ does not join $D$ in both executions or node $v$ joins $D$ in both executions in the same iteration.
Therefore, the two executions have the same partial dominating set $D$ after every iteration of stage $i$.
\end{lemma}
\begin{proof}
Couple the two executions by using the same sequence of coins for every candidate $v$, and let $J=\min\{k:\text{coin }k\text{ is successful}\}$.
If $d_i^J(v)\ge T_i$, then $v$ joins $D$ in iteration $J$ in both executions.
Otherwise, $d_i^J(v)<T_i$, and since residual degree is nonincreasing during a stage, $d_i^k(v)<T_i$ for every $k>J$.
Thus $v$ does not join $D$ in either execution.

Hence the same nodes join $D$ in every iteration, and the two executions have the same partial dominating set $D$ after every iteration of stage $i$.
\end{proof}

\subsubsection{Correctness of estimation}

The next lemma shows that, in Phase~1, the sampling estimates of residual degree yield a correct eligibility decision whenever the true residual degree is sufficiently far from the threshold. 

\begin{lemma}
\label{lem:est-decision-base}
Fix constants $\gamma\in(0,1/2)$ and $c_0\ge 1$.
For a sufficiently large constant $C=C(\gamma,c_0)$, with probability at least $1-n^{-c_0}$, the following holds   for every Phase~1 candidate $v$, every stage $i$ in which it is a candidate, and every iteration $k$ of that stage:
\[
  d_i^k(v)\ge(1+\gamma)T_i
  \;\Longrightarrow\;
  \widehat d(v)\ge T_i,\qquad\qquad 
  d_i^k(v)\le(1-\gamma)T_i
  \;\Longrightarrow\;
  \widehat d(v)<T_i,
\]
where $\widehat d(v)$ denotes the residual degree estimate that $v$ computes in iteration $k$, stage $i$.
\end{lemma}
\begin{proof}
Let $L=\lceil\log\Delta\rceil$. Consider a vertex $v$ which is a Phase~1 candidate in a stage $i<i^\star$. Consider an iteration $k$ in stage $i$.
Let $A$ be the nodes of $N^+(v)$ that are undominated at the beginning of stage $i$, and let $B\subseteq A$ be those still undominated immediately before iteration $k$. 
Thus $|A|=d_i^1(v)$ and $|B|=d_i^k(v)$. 

Our goal is to show that the fraction of estimators that are still undominated is close to $|B|/|A|$. 
Let $X$ be the number of estimators in iteration $k$, stage $i$, that belong to $A$ and $Y$ be the number of these estimators that belong to $B$.
Since we use $\widehat d(v)=(Y/X)|A|$ as the residual degree estimator, if $Y/X \sim |B|/|A|$, then $\widehat d(v) \sim |B|$.

\noindent
\textit{First, we view the estimator assignments sequentially.} Let $\mathcal{S}$ denote the $L^2$ stage--iteration pairs, ordered lexicographically. Although the algorithm assigns to each node one stage--iteration pair selected uniformly at random at the beginning (Line 2, Algorithm \ref{alg:base-awake}), for the purposes of the analysis, we may equivalently generate these assignments in each iteration, as follows.
Let $M$ denote the number of stage--iteration pairs in the suffix of $\mathcal{S}$ starting at (and including) the current pair $(i, k)$.
Imagine that at the start of iteration $k$ in stage $i$, every node not yet assigned a stage--iteration pair is assigned the current pair $(i, k)$, independently, with probability $1/M$.
A simple coupling argument in which the random choices made by the two processes are coupled together, can establish the equivalence.
Indeed, the probability that a node is assigned to any particular pair is $\prod_{r=M+1}^{L^2}(1-1/r)\cdot(1/M)=1/L^2$. 
Furthermore, the probability that a node has not been assigned prior to the current stage--iteration pair $(i, k)$ is $\prod_{r=M+1}^{L^2}(1-1/r)=M/L^2$.
Let $A'\subseteq A$ be the nodes assigned to stage--iteration pair $(i,k)$ or to a later pair, and let $B'=B\cap A'$.

\noindent
\textit{Next, we bound the sizes of $|A'|$ and $|B'|$.}
We know that $\mathbb{E}[|A'|] = |A|\cdot M/L^2$. 
Since $v$ is a candidate in stage $i$, $|A| \ge T_i$.
Furthermore, since $i < i^\star$, $T_i > CL^2\log n$.
Therefore, $\mathbb{E}[|A'|] > C \log n$ and using Chernoff bounds~\cite[Chapter~4]{mitzenmacher2005probability}, in the form of Corollary~\ref{cor:chernoff-whp}(1) in Appendix~\ref{app:chernoff}, we get that 
\begin{equation}
\label{eqn:APrimeBound}
\mathrm{Prob}[|A'| \in (1\pm\gamma/4)|A|\cdot (M/L^2)] \ge 1 - n^{-(c_0+5)}.
\end{equation}
Similarly, $\mathbb{E}[|B'|] = |B|\cdot M/L^2$. To bound the deviation of $|B'|$ from its expectation, we assume that $|B|\ge\gamma T_i/8$. 
An application of Chernoff bounds implies that 
\begin{equation}
\label{eqn:BPrimeBound}
\mathrm{Prob}[|B'| \in (1\pm\gamma/4)|B| \cdot (M/L^2)] \ge 1 - n^{-(c_0+5)}.
\end{equation}
We postpone considering the case $|B|<\gamma T_i/8$ until the end of the proof.

\noindent
\textit{Next we bound $X$ and $Y$ in terms of $|A'|$ and $|B'|$ respectively.}
Conditioned on all assignments in previous stage--iteration pairs, $A'$ and $B'$ are fixed.
According to our sequential view of selecting estimators in each iteration, each node of $A'$ is assigned to the current stage--iteration pair $(i, k)$ independently, with probability $1/M$. 
Thus $X$ counts the nodes of $A'$ assigned to the current stage--iteration pair, while $Y$ counts the nodes of $B'$ assigned to the current stage--iteration pair. 
Therefore $\mathbb{E}[X|A'] = |A'|/M$ and $\mathbb{E}[Y|A'] = |B'|/M$.
Let $\mathcal{E}_1$ denote the event $|A'| \in (1\pm\gamma/4)|A|\cdot (M/L^2)$, whose probability is bounded in (\ref{eqn:APrimeBound}).
We now bound $X$ conditioned on $\mathcal{E}_1$.
\begin{equation}
\mathrm{Prob}\left[X \in \left(1 \pm \gamma/4\right) \frac{|A'|}{M} \middle| \mathcal{E}_1\right] \ge 1 - \exp\left(-\left(\frac{\gamma}{4}\right)^2 \cdot \frac{|A'|}{M} \cdot \frac{1}{3}\right) \ge 1 - n^{-(c_0+5)}.
\end{equation}
The second inequality above is due to the conditioning on $\mathcal{E}_1$, which implies that
$|A'| \ge (1 - \gamma/4)|A|\cdot (M/L^2) > (1 - \gamma/4)C\log n$.
To obtain bounds on $Y$, as before, we assume $|B|\ge\gamma T_i/8$.
Let $\mathcal{E}_2$ denote the event 
$|B'| \in \left(1\pm\gamma/4\right)|B|\cdot (M/L^2)$, whose probability is bounded in (\ref{eqn:BPrimeBound}). We bound $Y$ conditional on $\mathcal{E}_2$ as follows:
\begin{equation}
\mathrm{Prob}\left[Y \in (1 \pm \gamma/4) \frac{|B'|}{M} \middle| \mathcal{E}_2\right] \ge 1 - \exp\left(-\left(\frac{\gamma}{4}\right)^2 \cdot \frac{|B'|}{M} \cdot \frac{1}{3}\right) \ge 1 - n^{-(c_0+5)}.
\end{equation}
The second inequality above is due to the conditioning on $\mathcal{E}_2$, which implies that
$|B'| \ge \left(1 - \frac{\gamma}{4}\right)|B|\cdot \frac{M}{L^2} \ge \left(1 - \frac{\gamma}{4}\right)\cdot \frac{\gamma T_i}{8} \cdot \frac{M}{L^2} > \left(1 - \frac{\gamma}{4}\right)\cdot\frac{\gamma}{8} \cdot C \log n.$

\noindent
\textit{Next we bound $X$ and $Y$ in terms of $|A|$ and $|B|$ respectively.}
\begin{eqnarray}
\mathrm{Prob}\left[X \in (1 \pm \gamma/4) \frac{|A|}{L^2}\right] & \ge & 
\mathrm{Prob}\left[X \in (1 \pm \gamma/4) \frac{|A|}{L^2} \middle| \mathcal{E}_1\right] \cdot \mathrm{Prob}[\mathcal{E}_1] \nonumber \\
& \ge & \mathrm{Prob}\left[X \in \frac{(1 \pm \gamma/4)}{(1 \pm \gamma/4)} \frac{|A'|}{M} \middle| \mathcal{E}_1\right] \cdot (1-n^{-(c_0+5)}) \nonumber \\
& \ge & \mathrm{Prob}\left[X \in \left(1 \pm \gamma/4\right) \frac{|A'|}{M} \middle| \mathcal{E}_1\right] \cdot (1-n^{-(c_0+5)})\nonumber \\
& \ge & (1 - n^{-(c_0+5)}) \cdot (1-n^{-(c_0+5)}) \nonumber \\
& \ge & 1 - n^{-(c_0+4)}. \label{eqn:XBound}
\end{eqnarray}
The second inequality above follows from (\ref{eqn:APrimeBound}).
As before, to bound $Y$, we assume $|B|\ge\gamma T_i/8$.
In this case, we use calculations very similar to the calculations above involving $X$ and $|A|$ to show that
\begin{equation}
\label{eqn:YBound}
\mathrm{Prob}\left[Y \in (1 \pm \gamma/4) \frac{|B|}{L^2}\right] \ge 1 - n^{-(c_0+4)}.
\end{equation}

\noindent
\textit{Finally, we translate the bounds on $X$ and $Y$ into decision probabilities.}
\begin{description}
\item[Case 1: $d_i^k(v)=|B|\ge(1+\gamma)T_i$.] Then,
\begin{eqnarray*}
\widehat d(v) & = & \frac{Y}{X} \cdot |A|\\
              & \ge & \frac{(1-\gamma/4) |B|}{L^2} \cdot \frac{L^2}{(1+\gamma/4)|A|}\cdot|A|\quad \text{with probability} \ge 1-n^{-(c_0+3)}\\
              & = & \frac{(1-\gamma/4)}{(1+\gamma/4)}\cdot |B|\\
              & \ge & \frac{(1-\gamma/4)}{(1+\gamma/4)} \cdot (1+\gamma) \cdot T_i\\
              & \ge & (1 - \gamma/2) \cdot (1+\gamma) \cdot T_i\\
              & \ge & T_i
\end{eqnarray*}
The second inequality above follows from (\ref{eqn:XBound}) and (\ref{eqn:YBound}).
\item[Case 2: $d_i^k(v)=|B|\le(1-\gamma)T_i$.] We consider two subcases depending on $|B|$.
\begin{itemize}
\item $|B|\ge\gamma T_i/8$: In this case, we use the upper bound on $Y$ in (\ref{eqn:YBound}), the lower bound on $X$ in (\ref{eqn:XBound}), and calculations very similar to Case 1 to show that $\widehat d(v) \le T_i$ with probability at least $1 - n^{-(c_0+3)}$.

\item $|B|< \gamma T_i/8$: 
Note that 
$$\mathbb{E}[Y] = \mathbb{E}[\mathbb{E}[Y||B'|]] = \frac{\mathbb{E}[|B'|]}{M} = \frac{|B|}{L^2} < \frac{\gamma T_i}{8L^2}.$$
Applying a one-sided Chernoff bound to $Y$ gives 
$$\mathrm{Prob}[Y\le(\gamma/4)T_i/L^2] \ge 1 - n^{-(c_0+4)}.$$ 
Here we use that $T_i > CL^2 \log n$.
Together with the lower bound on $X$ in (\ref{eqn:XBound}), and using similar calculations as in Case (1), this gives 
$$\widehat d(v)\le \frac{\gamma T_i}{4(1-\gamma/4)} < T_i$$
with probability at least $1 - n^{-(c_0+3)}$. The second inequality above requires $\gamma < 1/2$, as stipulated by the theorem statement.
\end{itemize}
\end{description}

Thus both claims in the theorem statement hold for the fixed triple $(v,i,k)$ with probability at least $1 - n^{-(c_0+3)}$. Since there are at most $nL^2\le n^3$ triples $(v,i,k)$, a union bound gives the claimed probability of at least $1-n^{-c_0}$.
\end{proof}

\subsubsection{Awake complexity}

\begin{lemma}\label{lem:phase1-awake-base}
In Algorithm~\ref{alg:base-awake}, each node is awake for $O(\log\Delta+\log\log\Delta)=O(\log\Delta)$ rounds during Phase~1.
\end{lemma}

\begin{proof}
\textbf{As a candidate}: Every node is awake for 2 rounds at the start of each stage $i$, for status exchange. If a node $v$ determines that it is a candidate for stage $i$, it wakes up during its first-success iteration $J_{v, i}$. This implies $O(\log\Delta)$ awake rounds during Phase~1. Furthermore, if a node joins $D$ in stage $i$, it wakes up an additional $O(\log\log\Delta)$ rounds in stage $i$, as prescribed by the virtual binary tree schedule $S_{J_{v,i}}$. 

\noindent
\textbf{As an estimator}: Every node $u$ wakes at its randomly assigned stage--iteration pair $(I_u, K_u)$ and also during the $O(\log\log \Delta)$ iterations in $S_{K_u}$ in stage $I_u$.  

\noindent
Thus nodes wake up for a total of $O(\log\Delta+\log\log\Delta)=O(\log\Delta)$ rounds.
\end{proof}

\subsubsection{Approximation analysis}

We start by stating the approximation guarantee provided by Algorithm \ref{alg:base-mds}, proved in \cite{grunau2020improved}. 

\begin{theorem}[\cite{grunau2020improved}, Theorem 3.1]
\label{thm:GrunauApproximation}
Let $G$ be some MDS instance and let $D$ denote the solution returned by Algorithm \ref{alg:base-mds}. Furthermore, let $\opt$ denote the value of an optimal fractional dominating set for $G$. Then, $\mathbb{E}[|D|] = O(\log \Delta) \cdot \opt$.
\end{theorem}

Algorithm \ref{alg:base-awake} faithfully mimics Algorithm \ref{alg:base-mds} except that it uses residual degree \textit{estimates}, rather than the actual residual degrees.
These residual degree estimates are approximately correct, as shown in Lemma \ref{lem:est-decision-base}.
The following lemma shows that the approximation guarantee in Theorem \ref{thm:GrunauApproximation} is maintained,
even using these approximate residual degree estimates.

\begin{lemma}[Approximate eligibility]
\label{lem:approx-elig-base}
Fix constants $0<a<1<b$ and a cutoff stage $h\in\{1,\dots,\lceil\log\Delta\rceil\}$. Consider a variant of the \textsc{BaseMDS} algorithm (Algorithm~\ref{alg:base-mds}) in which, for stages $i<h$, the exact eligibility test for being a candidate in stage $i$, $d_i^k(v)\ge T_i$, is replaced by any rule satisfying
\[
  d_i^k(v)\ge bT_i\;\Rightarrow\;v\text{ is a candidate},\qquad d_i^k(v)\le aT_i\;\Rightarrow\;v\text{ is not a candidate},
\]
with the count of eligible vertices in $N^{+}(u)\setminus D$ nonincreasing in $k$ while $u$ remains undominated; stages $i\ge h$ use the exact eligibility test. Then $D$ is an $O\!\left(\frac{b}{a}\cdot\log\Delta\right)$-approximate dominating set in expectation.
\end{lemma}

\begin{proof}
Let $U_i$ be the set of undominated vertices at the beginning of stage $i$, and let $D_i$ be the set of vertices selected during stage $i$.
The argument of Theorem~\ref{thm:pq-basemds} carries over with three modifications.

\textit{Residual invariant.} At the beginning of stage $i$, every vertex has residual degree at most $2bT_i$.
For $i=1$, this follows from $|N^{+}(v)|\le\Delta=2T_1$. For $i>1$, any vertex with residual degree at least $2T_i=T_{i-1}$ when $i-1\ge h$, or at least $2bT_i=bT_{i-1}$ when $i-1<h$, was eligible at the final probability-$1$ iteration of stage $i-1$ and joined $D$.

Let $x^\star$ be an optimal fractional dominating solution. Since every $u\in U_i$ satisfies $\sum_{v\in N^{+}(u)}x_v^\star\ge1$ and every vertex has at most $2bT_i$ undominated vertices in its closed neighborhood,
\[
 |U_i|
 \le \sum_v x_v^\star |N^{+}(v)\cap U_i|
 \le 2bT_i\cdot\opt .
\]

\textit{Per-stage first-cover bound.} For $u\in U_i$, let $X_u$ be the number of vertices selected in the iteration in which $u$ is first dominated, and let $X_u=0$ if $u$ remains undominated throughout the stage.
The first-cover argument in the proof of Theorem~\ref{thm:pq-basemds}, specialized to probability base $2$, gives $\e[X_u]\le4$.
The same argument applies because the number of eligible vertices in $N^{+}(u)\setminus D$ is nonincreasing during the stage.

\textit{Stage bound.}
Every vertex selected during stage $i$ has residual degree at least $aT_i$ at the moment of selection. Thus $aT_i|D_i|\le\sum_{u\in U_i}X_u$.
Taking expectations and using the two bounds above gives $\e[|D_i|]\le4|U_i|/(aT_i)\le(8b/a)\cdot\opt$.
Since $a$ and $b$ are constants, each stage contributes $O(\opt)$ vertices in expectation. Summing over the $\lceil\log\Delta\rceil$ stages gives $\e[|D|]=O(\log\Delta)\cdot\opt$.

Finally, the last stage uses the exact eligibility rule, $T_{\lceil\log\Delta\rceil}\le1$, and its final iteration has probability $1$. Hence every vertex is dominated when the algorithm terminates.
\end{proof}
 
Lemma~\ref{lem:est-decision-base} implies that we can apply the above lemma with $a=1-\gamma$, $b=1+\gamma$, and $h=i^\star$. 
Since $\gamma < 1/2$, we have $b/a = (1+\gamma)/(1-\gamma) < 3$, and so the approximation factor is $O(\log \Delta)$, matching the guarantee in Theorem \ref{thm:GrunauApproximation}.

\subsubsection{Analysis of the basic algorithm}

\begin{theorem}\label{thm:base-awake-formal}
Algorithm~\ref{alg:base-awake} computes an $O(\log\Delta)$-approximate dominating set in expectation in $O(\log^2\Delta)$ rounds, with each node awake in at most $O(\log\Delta\log\log n)$ rounds.
\end{theorem}

\begin{proof}
\textit{Rounds.} Phase~1 and Phase~2 together run $\lceil\log\Delta\rceil$ stages (split at $i^\star$), each with $\lceil\log\Delta\rceil$ iterations, each iteration executed in $O(1)$ rounds.
This is a total of $O(\log^2 \Delta)$ rounds.

\noindent
\textit{Awake rounds.} Phase~1 requires $O(\log\Delta)$ awake rounds per node according to Lemma~\ref{lem:phase1-awake-base}. Phase~2 is executed with no sleeping and it adds $O(\log\Delta\log\log n)$ awake rounds, for a total of $O(\log\Delta\log\log n)$ awake rounds.

\noindent
\textit{Approximation.} 
The \textsc{BaseMDS} provides an $O(\log \Delta)$ approximation in expectation (Theorem \ref{thm:GrunauApproximation}). Fixing $c_0 = 2$, Lemma \ref{lem:est-decision-base} and Lemma \ref{lem:approx-elig-base} together imply that with probability at least $1 - n^{-2}$, Algorithm \ref{alg:base-awake} provides an $O(\log \Delta)$ approximation in expectation, even though it uses residual degree estimates to decide which nodes are candidates in each stage. With probability at most $n^{-2}$, $|D|$ can be much larger than $\opt$, but since $|D| \le n$, the additional contribution of this rare, bad event to $\mathbb{E}[D]$ is just $o(1)$.
\end{proof}

\section{The \texorpdfstring{$(\stagebase,\iterbase)$}{(p,q)}-MDS algorithm}
\label{sec:basemds-qary}

We now present a simple and natural generalization of \textsc{BaseMDS}, that we call $(\stagebase,\iterbase)$-MDS.
This algorithm yields a tradeoff between the approximation factor and the round count: by Theorem~\ref{thm:pq-basemds}, $(\stagebase,\iterbase)$-MDS returns an $O(\stagebase\iterbase\log_\stagebase\Delta)$-approximate dominating set in $O(\log_\stagebase\Delta\cdot\log_\iterbase\Delta)$ rounds in expectation. 
Thus increasing either one or both of the two parameters $\stagebase$ and $\iterbase$ reduces the number of
rounds, while increasing the approximation factor.
This round-approximation tradeoff implies a tradeoff between the awake complexity and
approximation factor as well via the awake-efficient version, described in the next section.

In the $(\stagebase,\iterbase)$-MDS algorithm, $1<\stagebase\le\Delta$ is the ratio between consecutive thresholds and $1<\iterbase\le\Delta$ is the base of the inner-loop probabilities. The algorithm runs $\lceil\log_\stagebase\Delta\rceil$ outer stages with thresholds $T_i=\Delta/\stagebase^i$, and each stage has $\lceil\log_\iterbase\Delta\rceil$ iterations with probabilities $\probit_k=\min\{1,\iterbase^k/\Delta\}$; setting $\stagebase=\iterbase=2$ recovers \textsc{BaseMDS}. Algorithm~\ref{alg:pq-mds} presents the pseudocode.

\begin{algorithm}[t]
\caption{$(\stagebase,\iterbase)$-MDS}
\label{alg:pq-mds}
\begin{algorithmic}[1]
\Require parameters $1<\stagebase,\iterbase\le\Delta$; maximum closed-neighborhood size $\Delta$ known to all nodes.
\State $D\gets\emptyset$.
\For{$i=1,2,\dots,\lceil\log_\stagebase\Delta\rceil$} \Comment{stages}
\For{$k=1,2,\dots,\lceil\log_\iterbase\Delta\rceil$} \Comment{iterations}
    \For{each node $v$ in parallel}
      \State $d_i^k(v)\gets|\{u\in N^{+}(v):u\notin N^{+}(D)\}|$.
      \If{$v\notin D$ and $d_i^k(v)\ge T_i$} \Comment{$T_i = \Delta/\stagebase^i$}
        \State add $v$ to $D$ with probability $\probit_k = \min\{1,\iterbase^k/\Delta\}$.
      \EndIf
    \EndFor
  \EndFor
\EndFor
\State \Return $D$.
\end{algorithmic}
\end{algorithm}

The analysis below generalizes the approximation analysis of~\cite[Lemmas~3.2--3.3 and Theorem~3.1]{grunau2020improved}.
In the \textsc{BaseMDS} algorithm of~\cite{grunau2020improved},
Lemma~3.3 bounds by a constant the expected number of neighbors of any node $v$ added to the dominating set in the stage in which the element is first covered.
Lemma~3.2 uses this bound to then bound the expected number of nodes added to the dominating set in the first stage,
and Theorem~3.1 applies this argument to the residual instance at each stage.
We use the same charging structure, but extend it to arbitrary $1<\stagebase,\iterbase\le\Delta$. 
We obtain a bound of $O(\iterbase)$ on the expected number of nodes added to the dominating set in the first stage.
Using this, we show that each stage contributes $O(\stagebase\iterbase)\cdot\opt$ nodes in expectation to the dominating set, and summing over $\lceil\log_\stagebase\Delta\rceil$ stages, gives the $O(\stagebase\iterbase\log_\stagebase\Delta)$ approximation bound.
This argument leads to the following theorem.

\medskip
\noindent\textbf{Theorem~\ref{thm:pq-basemds}} (restated)\textbf{.} \textit{For every $1<\stagebase,\iterbase\le\Delta$, there is a randomized \congest{} algorithm that returns an $O(\stagebase\cdot \iterbase\cdot \log_\stagebase\Delta)$-approximate dominating set in expectation, in $O(\log_\stagebase\Delta\cdot\log_\iterbase\Delta)$ rounds.}

\begin{proof}
The round complexity is immediate: $\lceil\log_\stagebase\Delta\rceil$ outer stages, each with $\lceil\log_\iterbase\Delta\rceil$ iterations, each implementable in $O(1)$ \congest{} rounds. We prove the approximation. Let $U_i$ be the set of undominated vertices at the beginning of stage $i$ and $D_i$ the set of vertices selected during stage $i$.

\textit{Residual degree invariant.} At the beginning of every stage $i$, every vertex's residual degree has size at most $\stagebase T_i$. For $i=1$, this follows from $|N^{+}(v)|\le\Delta=\stagebase T_1$. For $i>1$, the final iteration of stage $i-1$ has $\probit_{\lceil\log_\iterbase\Delta\rceil}=1$, so any vertex not in $D$ with residual degree at least $T_{i-1}=\stagebase T_i$ is eligible at that iteration and joins $D$ with probability one. Since an optimum dominating set covers $U_i$ and each of its vertices covers at most $\stagebase T_i$ vertices of $U_i$,
\[
   |U_i|\le \stagebase T_i\cdot\opt.
\]

\textit{First-cover charge.} Fix a stage $i$ and $u\in U_i$. If $u$ is not dominated during stage $i$, set $X_u=0$; otherwise, with $k$ the first iteration in which $u$ becomes dominated, let $X_u$ be the number of vertices selected at iteration $k$ whose closed neighborhoods contain $u$. While $u$ is undominated, the count of eligible vertices in $N^{+}(u)\setminus D$ is nonincreasing in $k$ (no neighbor of $u$ has joined $D$ and residual degrees only decrease). By deferred decisions~\cite[\S1.3]{mitzenmacher2005probability}, it suffices to consider a fixed nonincreasing sequence $n_1\ge n_2\ge\cdots\ge 0$, where $n_k$ is the number of eligible vertices in $N^{+}(u)$ just before iteration $k$. Put $r_k=n_k\,\probit_k$ and $R_k=\sum_{j\le k}r_j$. Then $\e[X_u]\le\sum_k r_k\,e^{-R_{k-1}}$, where $e^{-R_{k-1}}$ upper-bounds the probability that no eligible neighbor was selected at an earlier iteration. Since $r_1\le n_1\,\iterbase/\Delta\le \iterbase$ and $r_k\le \iterbase\,r_{k-1}$ for $k\ge 2$,
\[
   \sum_k r_k\,e^{-R_{k-1}}\le \iterbase+\iterbase\sum_{k\ge 2}r_{k-1}\,e^{-R_{k-1}}\le \iterbase+\iterbase\sum_{k\ge 2}\int_{R_{k-2}}^{R_{k-1}}e^{-x}\,dx\le 2\iterbase,
\]
using $r\,e^{-(A+r)}\le\int_A^{A+r}e^{-x}\,dx$. Therefore $\e[X_u]\le 2\iterbase$.

\textit{Stage bound.} Every vertex selected at stage $i$ has residual degree at least $T_i$ at the moment of selection (it is eligible there), so $T_i\,|D_i|\le\sum_{u\in U_i}X_u$. Taking expectations,
\[
   \e[|D_i|\mid U_i]\le \frac{2\iterbase\,|U_i|}{T_i}\le 2\stagebase\iterbase\cdot\opt.
\]
Summing over the $\lceil\log_\stagebase\Delta\rceil$ stages gives $\e[|D|]=O(\stagebase\iterbase\log_\stagebase\Delta)\cdot\opt$.

Finally, the output is a dominating set: $T_{\lceil\log_\stagebase\Delta\rceil}\le 1$ and $\probit_{\lceil\log_\iterbase\Delta\rceil}=1$; any vertex still undominated before the last iteration has residual degree at least $1\ge T_{\lceil\log_\stagebase\Delta\rceil}$, hence is eligible and joins with probability one.
\end{proof}

\section{Awake-efficient \texorpdfstring{$(\stagebase,\iterbase_1,\iterbase_2)$}{(p,q1,q2)}-MDS}
\label{sec:mds-awake-pq}
\label{sec:mds-energy}

This section presents an awake-efficient implementation of the $(\stagebase,\iterbase)$-MDS algorithm of Section~\ref{sec:basemds-qary}.
Transforming $(\stagebase,\iterbase)$-MDS into an awake efficient algorithm follows along the lines of transforming \textsc{BaseMDS} into \textsc{BaseMDS-Awake} (Algorithm \ref{alg:base-awake}), with one important difference.
As in \textsc{BaseMDS-Awake}, we use two phases, with nodes sleeping a lot in Phase 1, but staying awake fully in Phase 2.
For the stages in Phase 1, we invoke $(\stagebase,\iterbase_1)$-MDS, whereas for Phase 2 we invoke 
$(\stagebase,\iterbase_2)$-MDS, with different values of $\iterbase_1$ and $\iterbase_2$.
For the tradeoff analysis in Section~\ref{subsec:homestage-qary}, useful settings 
include $\iterbase_1 = O(1)$ and $\iterbase_2 = \log \Delta/\log\log n$.
We call this the $(\stagebase, \iterbase_1, \iterbase_2)$-\textsc{MDS-Awake} algorithm; its pseudocode is given in Algorithm~\ref{alg:mds-energy}.
Section~\ref{subsec:refined-cutoff} then uses a different cutoff and cleanup parameter to complete the proof of Theorem~\ref{thm:homestage-mds}.

\begin{algorithm}[t]
\caption{$(\stagebase,\iterbase_1,\iterbase_2)$-\textsc{MDS-Awake}}
\label{alg:mds-energy}
\begin{algorithmic}[1]
\Require Parameters $1<\stagebase,\iterbase_1,\iterbase_2\le\Delta$; 
$\Delta$ known to all nodes; constant $C$ from Lemma~\ref{lem:est-decision-base}; wake schedules $\{S_k\}_{k=1}^{\lceil\log_{\iterbase_1}\Delta\rceil}$ from Lemma~\ref{lem:vbt-base}.
\vspace{1ex}
\State $D\gets\emptyset$;\, $i^\star\gets\min\{i:T_i\le C\lceil\log_\stagebase\Delta\rceil\lceil\log_{\iterbase_1}\Delta\rceil\log n\}$\, \Comment{$T_i=\Delta/\stagebase^i$}
\State Each $v$ draws $(I_v,K_v)\sim\mathrm{Unif}([\lceil\log_\stagebase\Delta\rceil]\times[\lceil\log_{\iterbase_1}\Delta\rceil])$.
\vspace{1ex}
\Statex \textbf{Phase~1: sampling-based estimation.}
\For{$i=1,\dots,i^\star-1$} \Comment{stages}
  \State \parbox[t]{\linewidth-\algorithmicindent}{\textbf{Round 1}: every node wakes; every $v\in D$ broadcasts this fact, and any node hearing such a broadcast from a neighbor marks itself dominated.}
  \State \parbox[t]{\linewidth-\algorithmicindent}{\textbf{Round 2}: every node wakes; every undominated node broadcasts this status, and each $v$ uses this information to compute $d_i^1(v)=|N^{+}(v)\setminus N^{+}(D)|$.}
  \State \parbox[t]{\linewidth-\algorithmicindent}{Each $v$ for which $d_i^1(v)\ge T_i$ is a candidate.}
  \State \parbox[t]{\linewidth-\algorithmicindent}{Each candidate $v$ computes first-success iteration $J_{v,i}\in[\lceil\log_{\iterbase_1}\Delta\rceil]$ by tossing a sequence of biased coins with $\probit_k=\min\{1,\iterbase_1^k/\Delta\}$.
  \Comment{Lemma~\ref{lem:first-success-sampling-base}}}
  \For{$k=1,\dots,\lceil\log_{\iterbase_1}\Delta\rceil$} \Comment{iterations}
    \State Every candidate $v$ with $J_{v,i}=k$ wakes up.
    \State Every candidate $w\in D$ that joined earlier in stage $i$ wakes up if $k\in S_{J_{w,i}}$.
    \State Every estimator $u$ with $I_u=i$ wakes up if $k\in S_{K_u}$ and $k\le K_u$.
    \State \parbox[t]{\linewidth-\algorithmicindent-\algorithmicindent}{\textbf{Round 1}: every awake candidate in $D$ informs neighbors that it is in $D$; any awake node hearing this marks itself dominated.}
    \State \parbox[t]{\linewidth-\algorithmicindent-\algorithmicindent}{\textbf{Round 2}: every $u$ with $(I_u,K_u)=(i,k)$ that was undominated at the start of stage $i$ informs all neighbors of its status.}
    \State \parbox[t]{\linewidth-\algorithmicindent-\algorithmicindent}{Each candidate $v$ computes its estimate $\widehat d(v)$ from the received statuses; if $\widehat d(v)\ge T_i$, $v$ joins $D$, and schedules wake-ups at $\{\ell\in S_k:\ell>k\}$.}
  \EndFor
\EndFor
\Statex \textbf{Phase~2: exact cleanup.}
\State Run stages $i^\star,\dots,\lceil\log_\stagebase\Delta\rceil$ of $(\stagebase,\iterbase_2)$-MDS (Algorithm~\ref{alg:pq-mds}) exactly on the residual instance, with every node awake.
\State \Return $D$.
\end{algorithmic}
\end{algorithm}

The helper lemmas for \textsc{BaseMDS-Awake} translate to $(\iterbase_1, \iterbase_2)$-\textsc{MDS-Awake} in a straightforward manner.
The first-success sampling lemma
(Lemma~\ref{lem:first-success-sampling-base}) extends directly to Algorithm~\ref{alg:mds-energy}, 
with $\lceil\log_{\iterbase_1}\Delta\rceil$ Phase~1 iterations and probabilities $\probit_k=\min\{1,\iterbase_1^k/\Delta\}$.
The next lemma is the parameterized counterpart of
Lemma~\ref{lem:est-decision-base}.

\begin{lemma}[Estimator decision]
\label{lem:est-decision}
Fix $2\le\stagebase,\iterbase_1\le\Delta$, $\gamma\in(0,1/2)$, and $c_0\ge 1$. For a sufficiently large constant $C=C(\gamma,c_0)$, with probability at least $1-n^{-c_0}$, 
the following holds simultaneously for every Phase~1 candidate $v$, every stage $i<i^\star$ in which it is a candidate, and every iteration $k\in[\lceil\log_{\iterbase_1}\Delta\rceil]$:
\[
  d_i^k(v)\ge(1+\gamma)T_i
  \;\Longrightarrow\;
  \widehat d(v)\ge T_i,\qquad
  d_i^k(v)\le(1-\gamma)T_i
  \;\Longrightarrow\;
  \widehat d(v)<T_i,
\]
where $\widehat d(v)$ denotes the estimate that $v$ would compute in iteration $k$.
\end{lemma}

\begin{proof}
The proof of Lemma~\ref{lem:est-decision-base} applies with $L^2$ replaced by $\lceil\log_{\stagebase}\Delta\rceil\lceil\log_{\iterbase_1}\Delta\rceil$, the number of stage--iteration pairs.

Fix a Phase~1 candidate $v$, a stage $i<i^\star$, and an iteration $k$. Let $A$ be the nodes of $N^+(v)$ that are undominated at the beginning of stage $i$, and let $B\subseteq A$ be those still undominated immediately before iteration $k$, as in the proof of Lemma~\ref{lem:est-decision-base}. Thus $|A|=d_i^1(v)$ and $|B|=d_i^k(v)$. Since each node chooses its estimator pair uniformly from $[\lceil\log_{\stagebase}\Delta\rceil]\times[\lceil\log_{\iterbase_1}\Delta\rceil]$, the sequential exposure argument in that proof applies unchanged after replacing $L^2$ by $\lceil\log_{\stagebase}\Delta\rceil\lceil\log_{\iterbase_1}\Delta\rceil$.

Since $v$ is a candidate, $|A|\ge T_i$, and since $i<i^\star$, $T_i>C\lceil\log_{\stagebase}\Delta\rceil\lceil\log_{\iterbase_1}\Delta\rceil\log n$. Hence the same Chernoff-bound argument shows that, with failure probability at most $n^{-(c_0+3)}$, $X\in(1\pm\gamma/4)|A|/(\lceil\log_{\stagebase}\Delta\rceil\lceil\log_{\iterbase_1}\Delta\rceil)$ and, whenever $|B|\ge\gamma T_i/8$, $Y\in(1\pm\gamma/4)|B|/(\lceil\log_{\stagebase}\Delta\rceil\lceil\log_{\iterbase_1}\Delta\rceil)$. Therefore, exactly as in Lemma~\ref{lem:est-decision-base}, $|B|\ge(1+\gamma)T_i$ implies $\widehat d(v)\ge T_i$, while $|B|\le(1-\gamma)T_i$ and $|B|\ge\gamma T_i/8$ imply $\widehat d(v)<T_i$.

It remains to consider $|B|<\gamma T_i/8$. In this case, $\mathbb{E}[Y]=|B|/(\lceil\log_{\stagebase}\Delta\rceil\lceil\log_{\iterbase_1}\Delta\rceil)<\gamma T_i/(8\lceil\log_{\stagebase}\Delta\rceil\lceil\log_{\iterbase_1}\Delta\rceil)$. The same one-sided Chernoff bound as in Lemma~\ref{lem:est-decision-base}, together with the lower bound on $X$, gives $\widehat d(v)\le\gamma T_i/(4(1-\gamma/4))<T_i$, where the last inequality uses $\gamma<1/2$.

Thus both implications hold for every fixed triple $(v,i,k)$ with failure probability at most $n^{-(c_0+3)}$. Since $\stagebase,\iterbase_1\ge2$ and $\Delta\le n$, both $\lceil\log_{\stagebase}\Delta\rceil$ and $\lceil\log_{\iterbase_1}\Delta\rceil$ are at most $\lceil\log_2 n\rceil$. Therefore, there are at most $n\lceil\log_2 n\rceil^2=O(n\log^2 n)$ triples $(v,i,k)$. A union bound gives the claimed probability of at least $1-n^{-c_0}$.
\end{proof}

The approximate-eligibility lemma is restated below for general parameters as Lemma~\ref{lem:approx-elig}.

\begin{lemma}[Approximate eligibility]
\label{lem:approx-elig}
Fix $1<\stagebase,\iterbase_1,\iterbase_2\le\Delta$, constants $0<a<1<b$, and an integer $h\in\{0,\dots,\lceil\log_\stagebase\Delta\rceil-1\}$. Suppose we invoke $(\stagebase,\iterbase_1)$-MDS in stages $i\le h$ and $(\stagebase,\iterbase_2)$-MDS in stages $i>h$. For stages $i\le h$, the exact eligibility test $d_i^k(v)\ge T_i$ is replaced by any rule satisfying
\[
  d_i^k(v)\ge bT_i\;\Rightarrow\;v\text{ is a candidate},\qquad d_i^k(v)\le aT_i\;\Rightarrow\;v\text{ is not a candidate},
\]
with the count of eligible vertices in $N^{+}(u)\setminus D$ nonincreasing in $k$ while $u$ remains undominated; stages $i>h$ use the exact rule. Then $D$ is an $O\!\bigl(\iterbase_1\stagebase h+\iterbase_2\stagebase(\lceil\log_\stagebase\Delta\rceil-h)\bigr)$-approximate dominating set in expectation.
\end{lemma}

\begin{proof}
The argument of Theorem~\ref{thm:pq-basemds} carries over with three modifications. 

\textit{Residual invariant.} At the beginning of stage $1$, every vertex has residual degree at most $\Delta=\stagebase T_1\le\stagebase bT_1$. Now consider a stage $i>1$. If $i-1>h$, then stage $i-1$ uses the exact eligibility rule, so any undominated vertex with residual degree at least $\stagebase T_i=T_{i-1}$ was eligible in the final probability-$1$ iteration of stage $i-1$ and joined $D$. If $i-1\le h$, then stage $i-1$ uses the approximate eligibility rule, so any undominated vertex with residual degree at least $\stagebase bT_i=bT_{i-1}$ was a candidate in the final probability-$1$ iteration of stage $i-1$ and joined $D$. Hence, at the beginning of stage $i$, every vertex has residual degree at most $\stagebase bT_i$. Since an optimum dominating set covers all undominated vertices, the number of undominated vertices at the beginning of stage $i$ is at most $\stagebase bT_i\cdot\opt$.

\textit{Per-stage first-cover bound.} The first-cover argument in the proof of Theorem~\ref{thm:pq-basemds} shows that, for any fixed vertex, the expected number of selected vertices whose closed neighborhoods contain it in the iteration when it is first dominated is at most $2\iterbase_1$ in a Phase~1 stage and at most $2\iterbase_2$ in a Phase~2 stage.

\textit{Stage bound.} Every vertex selected in stage $i$ has residual degree at least $aT_i$ at the moment of selection. Combining this fact with the residual and first-cover bounds above, the expected number of vertices selected in each Phase~1 stage is at most $2\iterbase_1\stagebase b/a\cdot\opt=O(\iterbase_1\stagebase)\cdot\opt$. The analogous bound with $\iterbase_2$ holds for each Phase~2 stage. Summing over the $h$ Phase~1 stages and the $\lceil\log_\stagebase\Delta\rceil-h$ Phase~2 stages gives the claimed bound.

Finally, since $h\le\lceil\log_\stagebase\Delta\rceil-1$, the final stage uses the exact eligibility rule. In the final iteration of that stage, the selection probability is $1$ and the threshold is at most $1$, so every vertex that is still undominated joins $D$. Therefore, the output is a dominating set.
\end{proof}

\noindent
Lemma~\ref{lem:phase1-awake} bounds the awake complexity of Phase 1 of Algorithm \ref{alg:mds-energy}.
It extends Lemma~\ref{lem:phase1-awake-base}
to the parameters $\stagebase,\iterbase_1,\iterbase_2$ by replacing $2$ throughout the proof with the appropriate parameter.
Lemma \ref{lem:homestage-cleanup} bounds the awake complexity of Phase 2.
\begin{lemma}\label{lem:phase1-awake}
In Algorithm~\ref{alg:mds-energy}, each node is awake for $O(\log_\stagebase\Delta + \log\log_{\iterbase_1}\Delta)$ rounds during Phase~1.
\end{lemma}
\begin{proof}
$O(1)$ candidate wake-ups per stage contribute $O(\lceil\log_\stagebase\Delta\rceil)$ across Phase~1; a single join contributes $O(\log\lceil\log_{\iterbase_1}\Delta\rceil)$ rebroadcast wake-ups in $S_{J_{v,i}}$; the estimator role contributes another $O(\log\lceil\log_{\iterbase_1}\Delta\rceil)$ wake-ups in $S_{K_u}$.
\end{proof}

\begin{lemma}\label{lem:homestage-cleanup}
Suppose $2\le\stagebase\le\log n$ and $2\le\iterbase_1\le\Delta$. In Algorithm~\ref{alg:mds-energy}, Phase~2 executes the remaining $O(\log_\stagebase\log n)$ stages of the $(\stagebase,\iterbase_2)$-variant of \textsc{BaseMDS} exactly. It uses $O(\log_\stagebase\log n\cdot\log_{\iterbase_2}\Delta)$ rounds and the same number of awake rounds per node.
\end{lemma}

\begin{proof}
Since $\stagebase,\iterbase_1\ge2$ and $\Delta\le n$, both $\lceil\log_\stagebase\Delta\rceil$ and $\lceil\log_{\iterbase_1}\Delta\rceil$ are $O(\log n)$. Hence, the definition of $i^\star$ gives $T_{i^\star}=\Delta/\stagebase^{i^\star}=O(\log^3 n)$. Taking logarithms base $\stagebase$ gives $i^\star\ge\log_\stagebase\Delta-O(\log_\stagebase\log n)$. Since $\lceil\log_\stagebase\Delta\rceil\le\log_\stagebase\Delta+1$, it follows that $\lceil\log_\stagebase\Delta\rceil-i^\star+1=O(1+\log_\stagebase\log n)$.
Since $\stagebase\le\log n$, we have $\log_\stagebase\log n\ge1$, and therefore
$1+\log_\stagebase\log n\le 2\log_\stagebase\log n.$
Thus,
\[
\lceil\log_\stagebase\Delta\rceil-i^\star+1
=O(\log_\stagebase\log n).
\]
Each remaining stage has $\lceil\log_{\iterbase_2}\Delta\rceil$ iterations executed in $O(1)$ rounds with every node awake.
\end{proof}

\subsection{The tradeoff theorem}
\label{subsec:homestage-qary}

For $2\le\stagebase\le\log n$, $2\le\iterbase_1\le\Delta$, and $1<\iterbase_2\le\Delta$, Lemmas \ref{lem:phase1-awake} and \ref{lem:homestage-cleanup} show that the awake complexity of $(\stagebase,\iterbase_1,\iterbase_2)$-MDS is
\begin{equation}
\label{eqn:generalAwake}
O(\log_\stagebase \Delta + \log\log_{\iterbase_1} \Delta + \log_{\stagebase}\log n\cdot\log_{\iterbase_2}\Delta).
\end{equation}
Using Lemma \ref{lem:approx-elig} with $h=i^\star-1$ and Lemma \ref{lem:homestage-cleanup}, we see that
Phase~1 contributes $O(\stagebase\iterbase_1\log_\stagebase\Delta)$ to the expected approximation, whereas Phase~2 contributes $O(\stagebase\iterbase_2\log_\stagebase\log n)$, for a total of
\begin{equation}
\label{eqn:generalApprox}
O(\stagebase\iterbase_1\log_\stagebase\Delta + \stagebase\iterbase_2\log_\stagebase\log n).
\end{equation}
Using the setting of $\stagebase = \alpha$, $\iterbase_1 = \iterbase_2 = 2$ in (\ref{eqn:generalAwake}) and (\ref{eqn:generalApprox}) and simplifying yields an awake complexity of $O(\log_{\alpha} \Delta \cdot \log\log n)$ and approximation factor of $O(\alpha\cdot(\log_\alpha \Delta + \log_\alpha\log n))$.

We now focus on the last term ``$\log_{\stagebase}\log n\cdot\log_{\iterbase_2}\Delta$'' in (\ref{eqn:generalAwake}).
This last term can be reduced by increasing $\iterbase_2$, until $\iterbase_2=\Theta(\log n)$, at which point the first and the last terms are asymptotically identical.
Furthermore, this increase will not worsen the approximation factor when the second term ``$\stagebase\iterbase_2\log_\stagebase\log n$'' in (\ref{eqn:generalApprox}) is less than the first term. This motivates setting $\iterbase_2=\max\{2,\log\Delta/\log\log n\}$, which equalizes the two terms asymptotically when $\log\Delta=\Omega(\log n)$.
With this setting of $\iterbase_2$, the tradeoff improves as $\log\Delta$ becomes large relative to $\log\log n$.
In the following theorem, we state both the basic tradeoff and improved tradeoffs for some larger values of $\Delta$ relative to $n$.
The theorem is proved by plugging in the above-specified values for $\stagebase$, $\iterbase_1$, and $\iterbase_2$ into (\ref{eqn:generalAwake}), (\ref{eqn:generalApprox}), and the round complexity of $O(\log_\stagebase\Delta\cdot\log_{\iterbase_1}\Delta+\log_\stagebase\log n\cdot\log_{\iterbase_2}\Delta)$.

\begin{theorem}\label{thm:homestage-mds-tradeoff-formal}
For any $1<\alpha\le\min\{\Delta,\log n\}$, there is an algorithm in the sleeping \congest{} model that computes an $O(\alpha\cdot(\log_\alpha \Delta + \log_\alpha\log n))$-approximate dominating set in expectation in $O(\log_\alpha\Delta\cdot(\log\Delta + \log\log n))$ rounds, with each node awake in at most $O(\log_\alpha \Delta\cdot\log\log n)$ rounds.
\begin{itemize}
\item When $\log\Delta\ge (\log\log n)^2$, this tradeoff improves to $O(\alpha\log_\alpha \Delta)$-approximation in expectation, $O(\log_\alpha\Delta\cdot\log\Delta)$ rounds and $O(\log_\alpha\Delta\cdot\frac{\log\log n}{\log\log \Delta} + \log\log\Delta)$ awake complexity.
\item When $\log\Delta\ge (\log n)^{\varepsilon}$ for some constant $\varepsilon \in (0,1)$, this tradeoff further improves to $O(\alpha\log_\alpha \Delta)$-approximation in expectation, $O(\log_\alpha\Delta\cdot\log\Delta)$ rounds, and $O(\log_\alpha\Delta + \log\log\Delta)$ awake complexity.
\end{itemize}
\end{theorem}

\begin{proof}
For $1<\alpha<2$, we run the algorithm with parameter $2$; all the claimed bounds continue to hold up to constant factors. Hence assume $2\le\alpha\le\min\{\Delta,\log n\}$. Since $\alpha\le\log n$, we have $\log_\alpha\log n\ge1$.

For the basic tradeoff, we instantiate Algorithm~\ref{alg:mds-energy} with $\stagebase=\alpha$ and $\iterbase_1=\iterbase_2=2$. Substituting these values into (\ref{eqn:generalApprox}) gives an expected approximation factor of $O(\alpha\cdot(\log_\alpha\Delta+\log_\alpha\log n))$. The round complexity is $O(\log_\alpha\Delta\cdot(\log\Delta+\log\log n))$, where we use $\log_\alpha\log n\cdot\log\Delta=\log_\alpha\Delta\cdot\log\log n$. Similarly, (\ref{eqn:generalAwake}) gives awake complexity $O(\log_\alpha\Delta+\log\log\Delta+\log_\alpha\Delta\cdot\log\log n)$. Since $\log_\alpha\Delta\ge1$ and $\log\log\Delta\le\log\log n$, this simplifies to $O(\log_\alpha\Delta\cdot\log\log n)$.

For the first improvement, suppose that $\log\Delta\ge(\log\log n)^2$. We instantiate Algorithm~\ref{alg:mds-energy} with $\stagebase=\alpha$, $\iterbase_1=2$, and $\iterbase_2=\max\{2,\log\Delta/\log\log n\}$. For sufficiently large $n$, we have $\iterbase_2=\log\Delta/\log\log n$. The assumed lower bound on $\log\Delta$ gives $\log\log\Delta\ge2\log\log\log n$, and therefore $\log\iterbase_2=\log\log\Delta-\log\log\log n\ge\frac{1}{2}\log\log\Delta$. Thus, $\log_{\iterbase_2}\Delta=O(\log\Delta/\log\log\Delta)$.

Phase~1 contributes $O(\alpha\log_\alpha\Delta)$ to the expected approximation factor, while Phase~2 contributes $O(\alpha\iterbase_2\log_\alpha\log n)=O(\alpha\log_\alpha\Delta)$. Thus, the expected approximation factor is $O(\alpha\log_\alpha\Delta)$. Phase~1 uses $O(\log_\alpha\Delta\cdot\log\Delta)$ rounds, while Phase~2 uses $O(\log_\alpha\Delta\cdot\log\log n/\log\log\Delta)$ rounds. Under the assumed lower bound on $\log\Delta$, the Phase~2 round complexity is absorbed by the Phase~1 round complexity. Thus, the total round complexity is $O(\log_\alpha\Delta\cdot\log\Delta)$.

By Lemmas~\ref{lem:phase1-awake} and~\ref{lem:homestage-cleanup}, the awake complexity is $O(\log_\alpha\Delta+\log\log\Delta+\log_\alpha\Delta\cdot\log\log n/\log\log\Delta)$. Since $\log\log\Delta\le\log\log n$, this simplifies to $O(\log_\alpha\Delta\cdot\log\log n/\log\log\Delta+\log\log\Delta)$.

Finally, suppose that $\log\Delta\ge(\log n)^\varepsilon$ for some constant $\varepsilon\in(0,1)$. For sufficiently large $n$, this condition also implies $\log\Delta\ge(\log\log n)^2$, so the first improvement applies. Furthermore, $\log\log\Delta\ge\varepsilon\log\log n$, and hence $\log\log n/\log\log\Delta=O(1)$. The awake complexity therefore simplifies to $O(\log_\alpha\Delta+\log\log\Delta)$.

Fix $c_0=2$. Compare Algorithm~\ref{alg:mds-energy} with an execution using the same random choices in which any estimator decision violating Lemma~\ref{lem:est-decision} is corrected; the corrected decisions remain nonincreasing. Lemma~\ref{lem:approx-elig} with $h=i^\star-1$, $a=1-\gamma$, and $b=1+\gamma$ gives the claimed expected approximation bounds for the corrected execution. By Lemma~\ref{lem:est-decision}, the two executions coincide with probability at least $1-n^{-2}$; otherwise, the output size is at most $n$. Hence, the expected output size of Algorithm~\ref{alg:mds-energy} is at most the claimed bound plus $n\cdot n^{-2}=1/n$, which is absorbed because $\opt\ge1$.
\end{proof}

\subsection{Achieving \texorpdfstring{$O(\log\Delta)$}{O(log Delta)} awake complexity}
\label{subsec:refined-cutoff}

We now complete the proof of Theorem~\ref{thm:homestage-mds}. Phase~1 of the basic implementation already uses $O(\log\Delta)$ awake rounds per node. To obtain the same bound for Phase~2, we consider a variant of Algorithm~\ref{alg:base-awake} with a cutoff that depends only on $\Delta$ and a faster exact cleanup.

Let $L=\lceil\log\Delta\rceil$ and $i^\star=\min\{i\in\{1,\dots,L\}:T_i\le CL^4\},$ where $C\ge1$ is a sufficiently large constant, chosen below. Phase~1 runs stages $1,\dots,i^\star-1$ exactly as in Algorithm~\ref{alg:base-awake}.
Phase~2 runs the remaining stages $i^\star,\dots,L$ of $(2,\iterbase_2)$-MDS (Algorithm~\ref{alg:pq-mds}) exactly on the residual instance, with every node awake, where $\iterbase_2=\max\{2,\log\Delta/\max\{1,\log\log\Delta\}\}.$

A Phase~1 stage may now have $T_i\le CL^2\log n$, so Lemma~\ref{lem:est-decision-base} does not provide the required guarantee for this cutoff.
The next lemma bounds the probability of an incorrect estimator decision for each fixed node, stage, and iteration.

\begin{lemma}[Estimator decision, refined cutoff]
\label{lem:est-decision-refined}
Fix a constant $\gamma\in(0,1/2)$. For a sufficiently large constant $C=C(\gamma)$, the following holds for every node $v$, every stage $i<i^\star$, and every iteration $k$ of stage $i$. With probability at least $1-\exp(-\Omega(C\log^2\Delta))$, if $v$ is a candidate for stage $i$ then
\[
  d_i^k(v)\ge(1+\gamma)T_i
  \;\Longrightarrow\;
  \widehat d(v)\ge T_i,\qquad\qquad
  d_i^k(v)\le(1-\gamma)T_i
  \;\Longrightarrow\;
  \widehat d(v)<T_i,
\]
where $\widehat d(v)$ denotes the residual degree estimate that $v$ would compute in iteration $k$, stage $i$.
\end{lemma}

\begin{proof}
Fix a node $v$, a stage $i<i^\star$, and an iteration $k$. We apply the proof of Lemma~\ref{lem:est-decision-base} to this fixed triple, using Lemma~\ref{lem:chernoff} directly instead of Corollary~\ref{cor:chernoff-whp}.
For a stage-$i$ candidate, each concentration bound in that proof has failure probability at most $\exp(-\Omega(T_i/L^2))$, where the constant in the $\Omega$ depends only on $\gamma$.
Since $i<i^\star$, we have $T_i>CL^4$, and hence $T_i/L^2>CL^2$.
A union bound over the constant number of concentration bounds gives both implications with probability at least $1-\exp(-\Omega(CL^2))=1-\exp(-\Omega(C\log^2\Delta))$.
We do not take a union bound over nodes, stages, or iterations.
\end{proof}

\begin{theorem}
\label{thm:base-awake-refined}
The above variant of Algorithm~\ref{alg:base-awake} computes an $O(\log\Delta)$-approximate dominating set in expectation in $O(\log^2\Delta)$ rounds, with each node awake in at most $O(\log\Delta)$ rounds.
\end{theorem}

\begin{proof}
Let $\eta$ denote the uniform failure-probability bound in Lemma~\ref{lem:est-decision-refined}. We assume that $\Delta$ is larger than a sufficiently large constant; otherwise, the algorithm uses $O(1)$ rounds, and since $|D|\le n\le\Delta\cdot\opt$ (shown below), all three bounds follow.

\textit{Rounds.} Phase~1 runs fewer than $L$ stages, each with $L$ iterations, each iteration executed in $O(1)$ rounds; this is $O(\log^2\Delta)$ rounds.
By the definition of $i^\star$, $T_{i^\star}=\Delta/2^{i^\star}\le CL^4$, so $i^\star\ge\log\Delta-\log(CL^4)$. Since $L\le\log\Delta+1$, Phase~2 runs $L-i^\star+1=O(\log\log\Delta)$ stages.
Since $\Delta$ is large, $\iterbase_2=\log\Delta/\log\log\Delta$ and $\log\iterbase_2=\log\log\Delta-\log\log\log\Delta\ge\frac{1}{2}\log\log\Delta$, so each Phase~2 stage has $\lceil\log_{\iterbase_2}\Delta\rceil=O(\log\Delta/\log\log\Delta)$ iterations, each executed in $O(1)$ rounds.
Thus Phase~2 uses $O(\log\Delta)$ rounds, and the total is $O(\log^2\Delta)$ rounds.

\noindent
\textit{Awake rounds.} Phase~1 requires $O(\log\Delta)$ awake rounds per node according to Lemma~\ref{lem:phase1-awake-base}, whose proof does not depend on the cutoff. Phase~2 is executed with no sleeping and it adds $O(\log\log\Delta)\cdot O(\log\Delta/\log\log\Delta)=O(\log\Delta)$ awake rounds, for a total of $O(\log\Delta)$ awake rounds.

\noindent
\textit{Approximation.} Let $U_i$ be the set of undominated vertices at the beginning of stage $i$, and let $D_i$ be the set of vertices selected during stage $i$.
We call an estimator decisions correct if it satisfies the two implications of Lemma~\ref{lem:est-decision-refined}. We account for incorrect decisions in expectation.
Let $x^\star$ be an optimal fractional dominating solution. Since every vertex $u$ satisfies $\sum_{v\in N^+(u)}x_v^\star\ge1$ and every closed neighborhood has at most $\Delta$ vertices,
\[
  n\le\sum_u\sum_{v\in N^+(u)}x_v^\star
  =\sum_v x_v^\star|N^+(v)|
  \le\Delta\cdot\opt.
\]
We choose $C$ large enough that $\Delta^3L\eta\le1$.

\textit{Residual degree invariant.} We show that $\e[|U_i|]=O(T_i)\cdot\opt$ for every stage $i\le i^\star$. For $i=1$, $|U_1|=n\le\Delta\cdot\opt=2T_1\cdot\opt$.
Now let $i<i^\star$ and fix a vertex $v$. If $d_{i+1}^1(v)\ge(1+\gamma)T_i$, then $d_i^k(v)\ge(1+\gamma)T_i$ for every iteration $k$ of stage $i$, since residual degrees are nonincreasing. Hence $v$ was a candidate for stage $i$. Since its residual degree is still positive, $v$ did not join $D$ during stage $i$, so its estimate in its first-success iteration $J_{v,i}$ violated Lemma~\ref{lem:est-decision-refined}.
A union bound over the $L$ possible values of $J_{v,i}$ bounds the probability of such an incorrect decision by $L\eta$. Since $d_{i+1}^1(v)\le\Delta$, we get $\e[d_{i+1}^1(v)]\le(1+\gamma)T_i+\Delta L\eta\le(1+\gamma)T_i+1$. As in the proof of Lemma~\ref{lem:approx-elig-base},
\[
  \e[|U_{i+1}|]
  \le \sum_v x_v^\star\,\e[d_{i+1}^1(v)]
  \le \bigl((1+\gamma)T_i+1\bigr)\cdot\opt
  = O(T_{i+1})\cdot\opt .
\]

\textit{First-cover charge.} Fix a Phase~1 stage $i$ and a vertex $u$. If $u\in U_i$ is dominated during stage $i$, let $X_u$ be the number of vertices in $N^+(u)$ selected in the iteration in which $u$ is first dominated. Otherwise, let $X_u=0$. Thus $X_u\le\Delta$.
Let $\mathcal{F}_u$ be the event that the estimator decisions for every stage-$i$ candidate $v\in N^+(u)$ and every iteration $k$ are correct. Since $|N^+(u)|\le\Delta$, a union bound gives $\mathrm{Prob}[\overline{\mathcal{F}_u}]\le\Delta L\eta$.
Compare the algorithm with an execution using the same random choices in which every incorrect estimator decision of a candidate in $N^+(u)$ during stage $i$ is corrected, as in the proof of Theorem~\ref{thm:homestage-mds-tradeoff-formal}. Let $\widetilde X_u$ denote its first-cover count. The two executions have the same set $U_i$, and on $\mathcal{F}_u$ they coincide throughout stage $i$. In the corrected execution, candidates in $N^+(u)$ satisfy both threshold implications, so the first-cover argument used in Lemma~\ref{lem:approx-elig-base} gives $\e[\widetilde X_u]\le4\cdot\mathrm{Prob}[u\in U_i]$. Since $X_u\le\Delta$,
\[
  \e[X_u]\le\e[\widetilde X_u]+\Delta\cdot\mathrm{Prob}[\overline{\mathcal{F}_u}]\le 4\cdot\mathrm{Prob}[u\in U_i]+\Delta^2L\eta .
\]

\textit{Stage bound.} Let $i$ be a Phase~1 stage. A selected vertex whose estimator decision in its first-success iteration is correct has residual degree greater than $(1-\gamma)T_i$ immediately before selection, and it is counted in $X_u$ for every vertex $u$ that it dominates for the first time. Thus at most $\sum_u X_u/((1-\gamma)T_i)$ such vertices are selected during stage $i$. Every other selected vertex has an incorrect estimator decision in its first-success iteration; as in the residual degree invariant, there are at most $nL\eta$ such vertices in expectation. Therefore,
\[
  \e[|D_i|]
  \le \frac{4\,\e[|U_i|]+n\Delta^2L\eta}{(1-\gamma)T_i}+nL\eta
  = O(\opt),
\]
using $\e[|U_i|]=O(T_i)\cdot\opt$, $n\le\Delta\cdot\opt$, and $\Delta^3L\eta\le1$. Summing over the fewer than $L$ Phase~1 stages, Phase~1 contributes $O(\log\Delta)\cdot\opt$ in expectation.

\textit{Phase~2.} At the beginning of Phase~2, $\e[|U_{i^\star}|]=O(T_{i^\star})\cdot\opt$ by the residual invariant. Phase~2 uses the exact eligibility rule, so the first-cover charge and the stage bound in the proof of Theorem~\ref{thm:pq-basemds} give $\e[|D_{i^\star}|]\le2\iterbase_2\,\e[|U_{i^\star}|]/T_{i^\star}=O(\iterbase_2)\cdot\opt$. The final iteration of stage $i^\star$ has probability $1$, so the residual degree invariant in that proof holds from stage $i^\star+1$ on, and every later stage also contributes $O(\iterbase_2)\cdot\opt$ in expectation. Since Phase~2 has $O(\log\log\Delta)$ stages and $\iterbase_2=O(\log\Delta/\log\log\Delta)$, it contributes $O(\log\Delta)\cdot\opt$ in expectation. Together with Phase~1, $\e[|D|]=O(\log\Delta)\cdot\opt$.

Finally, the last stage uses the exact eligibility rule, $T_L\le1$, and its final iteration has probability $1$. Hence every vertex is dominated when the algorithm terminates.
\end{proof}

\section{Conclusions}
\label{sec:conclusion}
We have shown that the awake complexity of computing an approximate MDS can be much smaller than the round complexity. Specifically, our algorithm computes an $O(\log\Delta)$-approximate MDS in expectation in $O(\log^2\Delta)$ rounds, with $O(\log\Delta)$ awake complexity. Thus we reduce awake complexity while preserving the round complexity and approximation guarantee.

Several questions remain open. Perhaps the most basic question raised by our work is whether we can achieve an $O(\log \Delta)$-approximate MDS with $o(\log \Delta)$ awake complexity and $O(\text{poly}(\log n))$ round complexity in the sleeping \congest{} model.
This problem seems to be open even in the sleeping \local{} model.

Keeping the overall structure of our algorithm, ensuring an $o(\log \Delta)$ awake complexity would imply that nodes don't wake up every stage and in fact sleep through many stages.
One natural idea along these lines is the following: suppose a node $v$ is awake in stage $i$, but has not been added to the dominating set in stage $i$. Node $v$ can use its residual degree at the end of stage $i$ to pick a target stage to wake up in. Node $v$ can then sleep until the start of the target stage. 
For example, if a node has very few undominated neighbors, it can potentially sleep for many stages. 
Unfortunately, adding this type of sleep schedule does not improve the awake complexity in the worst case.
For the counterexample in Appendix~\ref{app:worst-candidate}, a large fraction of the nodes wake up in many stages. This happens because a node may determine its target wake up stage based on its current residual degree, but by the time it has woken up, its residual degree has fallen further, making this a ``wasted'' wake up. 
In fact, in our counterexample, even the average awake complexity is $\Theta(\log \Delta)$. 
Designing an $O(\log \Delta)$-approximate MDS in $o(\log \Delta)$ awake rounds seems to be a challenging open problem requiring fundamentally new ideas. 

It is also natural to ask whether allowing a larger round complexity can lead to lower awake complexity. Our current techniques do not give such a tradeoff, and we do not know whether an $O(\log\Delta)$-approximation with $o(\log\Delta)$ awake complexity can be obtained in the sleeping \congest{} model even if polynomially many rounds are allowed.

Another open question is whether we can show lower bounds on the awake complexity of computing an approximate MDS? There are well-known round complexity lower bounds for MDS approximation due to Kuhn, Moscibroda, and Wattenhofer~\cite{kuhn2016local}. Can the indistinguishability technique used in these round complexity lower bound proofs be used to obtain awake complexity lower bounds as well? The actual lower bound on the awake complexity of computing an $O(\log \Delta)$-approximate MDS may be much smaller than the corresponding round complexity lower bound.

\bibliography{refs}
\appendix

\section{Chernoff bounds}
\label{app:chernoff}

We use the following standard forms of the Chernoff bound  (Chapter~4 of Mitzenmacher and Upfal~\cite{mitzenmacher2005probability}).

\begin{lemma}[Chernoff bounds]
\label{lem:chernoff}
Let $X_1, X_2, \ldots, X_N$ be independent random variables taking values in $\{0,1\}$, let $X = \sum_{j=1}^{N} X_j$, and let $\mu = \mathbb{E}[X]$.
\begin{enumerate}
\item For every $\delta \in (0,1)$, $\mathrm{Prob}[X \notin (1\pm\delta)\mu] \le 2\exp\left(-\delta^2\mu/3\right)$.
\item For every $\mu_H \ge \mu$ and every $\delta \in (0,1]$, $\mathrm{Prob}[X \ge (1+\delta)\mu_H] \le \exp\left(-\delta^2\mu_H/3\right)$.
\end{enumerate}
\end{lemma}

Part (1) is the two-sided bound we use when we know the expectation.
Part (2) is the one-sided bound we use when we only have an upper bound $\mu_H$ on the expectation, which is the situation in the case $|B| < \gamma T_i/8$ in the proof of Lemma~\ref{lem:est-decision-base}.
In that proof we apply these bounds to sums whose expectation, or whose upper bound $\mu_H$, is at least a constant fraction, depending only on $\gamma$, of $C\log n$, and we want a failure probability of at most $n^{-c}$ for a constant $c$ of our choosing.
The following consequence of Lemma~\ref{lem:chernoff} is the form we use; the constant $C$ of Lemma~\ref{lem:est-decision-base} is chosen large enough in terms of $\gamma$ and $c_0$ to absorb both these fractions and the constant factors in the two bounds above.

\begin{corollary}
\label{cor:chernoff-whp}
Fix constants $\delta \in (0,1)$ and $c > 0$, and let $X$ be as in Lemma~\ref{lem:chernoff}. There is a constant $C = C(\delta,c)$ such that the following hold for every $n \ge 2$.
\begin{enumerate}
\item If $\mathbb{E}[X] \ge C\log n$, then $\mathrm{Prob}[X \in (1\pm\delta)\mathbb{E}[X]] \ge 1 - n^{-c}$.
\item If $\mu_H \ge \mathbb{E}[X]$ and $\mu_H \ge C\log n$, then $\mathrm{Prob}[X \le 2\mu_H] \ge 1 - n^{-c}$.
\end{enumerate}
\end{corollary}

\begin{proof}
For (1), Lemma~\ref{lem:chernoff}(1) bounds the failure probability by $2\exp(-\delta^2 C\log n/3)$, which is at most $n^{-c}$ once $C$ is a large enough constant in terms of $\delta$ and $c$.
For (2), Lemma~\ref{lem:chernoff}(2) applied with $\delta = 1$ bounds the failure probability by $\exp(-\mu_H/3) \le \exp(-C\log n/3)$, which is again at most $n^{-c}$ for a large enough constant $C$.
\end{proof}
\section{Counterexample}
\label{app:worst-candidate}
Here we show an input graph $G$ where a constant fraction of the nodes will need to wake up in all but the first stage of the \textsc{BaseMDS} algorithm. We consider an awake-efficient implementation of the \textsc{BaseMDS} algorithm where each node will decide to wake up in the stage where it becomes a candidate according to its current residual degree. We ignore the energy cost to accurately estimate the residual degree. Assume for the sake of clarity that $\Delta=2^{L+1}$ where $L$ is the number of stages in \textsc{BaseMDS}. This means that for all $0 \le i \le L$, the thresholds for becoming a candidate in stage $i$, $T_i=\Delta/2^i$ are integers (and $T_{L}=2$). We construct a graph $G$ of maximum degree at most $\Delta$ in which every vertex in a core set \vcore{} is a candidate in at least $L$ stages of the $(\stagebase,\iterbase)$-MDS algorithm.  Let $|\vcore{}| = \Delta/2 - 2$.

For each $i=2,\dots,L$, create two blocks $B_i, Q_i$ each having $T_i$ vertices and a controller vertex $c_i$. Connect every vertex in \vcore{} to every vertex in $B_i$, and connect $c_i$ to every vertex in $B_i \cup Q_i$. We add another controller vertex $c_1$ that is adjacent to all nodes in \vcore{} and has additional $\Delta/2$ neighbors $Q_1$. There are no other edges in $G$. The number of vertices in $G$ is $1 + |\vcore| + |Q_1| + \sum_{i=2}^{L}{(1 + |B_i| + |Q_i|)} = \Delta - 1 + \sum_{i=2}^{L}{ (1 + \Delta/2^{i-1})} \le 4\Delta$. The following diagram illustrates one block of the construction.

\begin{figure}
    \centering
    \includegraphics[width=0.9\textwidth]{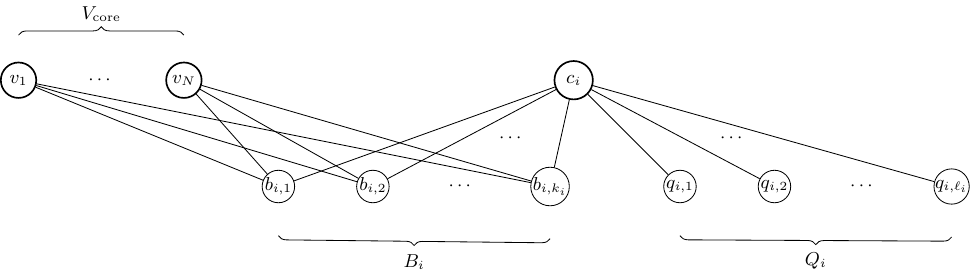}
    \caption{One part of our counterexample graph $G$. Each core vertex is adjacent to all vertices in $B_i$, and the controller $c_i$ is adjacent to all vertices in $B_i\cup Q_i$. This figure only describes the construction for $i>1$.}
    \label{fig:counterexample}
\end{figure}

We show that the maximum degree in $G$ is indeed at most $\Delta$. Every $v \in V_{\mathrm{core}}$ has degree
$1 + \sum_{i=2}^{L}|B_i|
= 1 + \sum_{i=2}^{L} \Delta/2^i
= \Delta/2 - 1$. Controller $c_1$  has degree $|V_{\mathrm{core}}| + |Q_1| = \Delta - 2$ and every controller $c_i$ for $i \ge 2$ has degree $|B_i|+|Q_i|= 2 T_i \le \Delta/2^{i-1}$. Each vertex in $B_i$ has degree $|V_{\mathrm{core}}| + 1 = \Delta/2 - 1$, and each vertex in $Q_i$ has degree $1$.

The degrees of each node are set up in such a way that only $c_i$ is a candidate in stage $i$, vertices in $B_i$ are candidates in the second stage, vertices in \vcore{} are candidates in the third stage, and vertices in $Q_i$ may be candidates only in the last stage, so all these vertices will sleep until these stages. Since $c_1$ is the only candidate in stage one, it will remain a candidate until it joins the dominating set in stage one, which is guaranteed to happen, since in the last iteration of each stage, the candidates join with probability 1. Once $c_1$ joins the dominating set, the vertices in \vcore{} are dominated and their residual degree is decreased by $2$. Whereas the residual degree of nodes in $B_i$ reduces to at most $2$ as most of their neighbors in \vcore{} are dominated. Therefore, nodes in $B_i$ wake up in the second stage, notice that their residual degree has fallen, and sleep until the last stage. This means that, $c_2$ is the only candidate in the second stage, and similarly will join the dominating set in some iteration of the second stage. Now, when nodes in \vcore{} wake up in the third stage, they find their residual degree has decreased to $\sum_{i=3}^{L}|B_i|= \sum_{i=3}^{L}(T_i)= (\Delta/8) (1 - 1/2^{L-2})$. This makes the \vcore{} nodes ineligible to be candidates in the third stage, but they are eligible to be candidates in the fourth stage. Therefore, all nodes in \vcore{} sleep until the fourth stage. This means that, $c_3$ is the only candidate in the third stage, and it will join the dominating set. This process will repeat in all future stages.

To summarize, in stage $i$, vertices in $\vcore{}$ and $c_i$ will wake up, the \vcore{} vertices will realize that their residual degree has fallen below the threshold for stage $i$, but it is above the threshold for stage $i+1$. Therefore, vertices in \vcore{} wake up in stage $i$ and decide to sleep until stage $i+1$. This allows $c_i$ to join the dominating set in stage $i$, which prevents vertices in \vcore{} from becoming candidates in stage $i+1$, which they will only figure out when they wake up in stage $i+1$.

Therefore, all vertices in \vcore{} wake up at least once in almost all the stages, which means that the (worst-case and average) awake complexity is $\Omega(L)$.

\end{document}